\documentclass[sigconf]{acmart}
\setcopyright{none} 

\usepackage[utf8]{inputenc} 
\usepackage[T1]{fontenc}
\usepackage{graphicx}
\usepackage{amsfonts}
\usepackage{color}
\usepackage{float}
\usepackage{verbatim}

\usepackage{amssymb}
\usepackage{amsthm}
\graphicspath{{fig/}}
\renewcommand{\paragraph}[1]{\medskip\noindent\textbf{#1}}

\newcommand{\alp}{\alpha}
\newcommand{\bet}{\beta}
\newcommand{\F}[2]{F_{#1,#2}} 
\newcommand{\A}[2]{A_{#1,#2}}
\renewcommand{\H}[2]{H_{#1,#2}}
\newcommand{\M}[2]{J_{#1,#2}}

\renewcommand{\P}[1]{P\left(#1\right)}
\newcommand{\V}[2]{W_{#1,#2}}
\newcommand{\Vr}[1]{I_{#1}}
\newcommand{\X}[2]{V_{#1,#2}}

\newcommand{\Y}[2]{Y_{#1,#2}}

\newcommand{\f}[1]{f_{#1}}
\newcommand{\h}[1]{h(#1)}
\newcommand{\bl}[1]{b_{#1}}
\newcommand{\expect}[1]{\mathbb{E}\left[#1\right]}
\newcommand{\ea}[0]{e^{-\alp}}
\newcommand{\ex}[1]{e^{#1}}
\newcommand{\wequ}[0]{{\color{white}=}}



\author{Jing Li, Dongning Guo, and Ling Ren\\
JingLi2015@u.northwestern.edu, dGuo@northwestern.edu, renling@illinois.edu}

\newcommand{\DG}[1]{{#1}}

\newcommand{\jing}[1]{{#1}}

\begin{document}
\title{Close Latency--Security Trade-off for the Nakamoto Consensus} 

\begin{abstract}
Bitcoin is a peer-to-peer electronic cash system invented by Nakamoto in 2008. While it has attracted much research interest, its exact latency and security properties remain open. Existing analyses provide security and latency (or confirmation time) guarantees that are too loose for practical use. In fact the best known upper bounds are several orders of magnitude larger than a lower bound due to a well-known private-mining attack. This paper describes a continuous-time model for blockchains and develops a rigorous analysis that yields close upper and lower bounds for the latency--security trade-off.  For example, when the adversary controls 10\% of the total mining power and the block propagation delays are within 10 seconds, a Bitcoin block is secured with less than $10^{-3}$ error probability if it is confirmed after four hours, or with less than $10^{-9}$ error probability if confirmed after ten hours. These confirmation times are about two hours away from their corresponding lower bounds. To establish such close bounds, the blockchain security question is reduced to a race between the Poisson adversarial mining process and a renewal process formed by a certain species of honest blocks. The moment generation functions of relevant renewal times are derived in closed form. The general formulas from the analysis are then applied to study the latency--security trade-off of several well-known proof-of-work longest-chain cryptocurrencies.  Guidance is also provided on how to set parameters for different purposes.
\end{abstract}




\maketitle

\section{Introduction}

Bitcoin was invented by Nakamoto~\cite{nakamoto2008bitcoin} in 2008 as a peer-to-peer electronic cash system.  It builds a distributed ledger commonly referred to as a blockchain.
The Bitcoin blockchain is a 
sequence of transaction-recording blocks which begins with a genesis block, and chains every subsequent block to a parent block using a cryptographic hash. 
Producing a new block requires proof-of-work (PoW) \emph{mining}: a nonce must be included such that the block's hash value satisfies a difficulty requirement.
An honest miner follows the longest-chain rule, i.e., it always tries to mine a block at the maximum height.
Miners form a peer-to-peer network to inform each other of newly mined or received blocks.

A block (and the transactions within) cannot be immediately confirmed upon its inclusion in the blockchain due to a phenomenon called \emph{forking}.
Different blocks may be mined and published at around the same time, so {different} honest miners may extend different blockchains.
Forking may also occur as a result of adversarial miners deviating from the longest-chain rule. 
Forking presents opportunities for double spending, which may happen if a transaction buried in a longest fork at one time is not included in another fork that later overtakes the former fork. 
In fact, unlike classic Byzantine fault tolerant protocols, 
the Bitcoin protocol only admits probabilistic guarantees. 
The latency (or confirmation time) of a block in Nakamoto-style consensus protocols (Nakamoto consensus for short) depends on the desired security level.

The goal of this work is to obtain close estimates of the fundamental latency--security trade-off for the Nakamoto consensus.
While the 
consensus protocol is simple and elegant, 
a rigorous analysis of its security
is very challenging.
The original Bitcoin white paper~\cite{nakamoto2008bitcoin} 
only analyzes 
a single specific attack, called the {\em private attack}, which is to mine an adversarial fork in private.
Nakamoto showed that the probability the adversary's private fork overtakes the main blockchain vanishes exponentially with the latency (see also~\cite{rosenfeld2014analysis}).
It is not until six years later that Garay et al.~\cite{garay2015bitcoin} provided the first 
proof that the Nakamoto consensus is secure against all possible attacks.
One {major} limitation of~\cite{garay2015bitcoin} 
{is that their} \emph{round-based lock-step synchrony} model
essentially abstracts away block propagation delays.
Several follow-up works~\cite{pass2017analysis,pass2017rethinking,kiffer2018better,bagaria2019prism} have extended the analysis to the $\Delta$-synchrony model in which the rounds in which different honest miners observe the same block may differ by
up to a known upper bound $\Delta$. 

So far, most existing analyses against all possible attacks~\cite{garay2015bitcoin, garay2017bitcoin,pass2017analysis,pass2017rethinking,kiffer2018better,bagaria2019prism} (including a few concurrent and follow-up works~\cite{bagaria2019proof,niu2019analysis,dembo2020everything,gazi2020tight}) focus on establishing asymptotic trends using the big $O(\cdot)$ or big $\Omega(\cdot)$ notation.
These results are not concrete bounds for any given security level because the latency implied could be hours, days, or even years depending on the unknown constants.
{The}
constants in some of these asymptotic results {have been worked out in}~\cite{
li2020liveness} {to reveal that} the latency upper bounds {are} 
several orders of magnitude higher than the best known lower bounds. 
Thus, despite their theoretical value, existing analyses of the Nakamoto consensus provide little guidance on the actual confirmation time, security guarantees, or parameter selection in practice. 

In this paper, we explicitly and closely 
characterize the trade-off between latency and security for Nakamoto-style PoW consensus protocols.
For Bitcoin,
the latency results we prove are typically within a few hours to 
lower bounds from the private attack.
The gap remains nearly constant at
different security levels, and is hence proportionally insignificant for high security levels. 
For example, with a 10\% adversarial mining power, a mining rate of one block every 10 minutes, and a maximum block propagation delay of 10 seconds, a block in the Nakamoto consensus is secured with $10^{-3}$ error probability if it is confirmed four hours after it is mined, or with $10^{-9}$ error probability if it is confirmed ten hours after it is mined.
As a reference, 
one must wait for at least two hours or six hours and forty minutes
before confirming for $10^{-3}$ and$10^{-9}$ security levels, respectively.
In contrast, the best prior bounds under essentially the same setting 
put the latency guarantees at thousands of hours or more~\cite{bagaria2019prism,li2020liveness, li2020continuous}.

Since Bitcoin's rise to fame, numerous altcoins and Bitcoin hard forks have adopted the Nakamoto consensus protocol with very different parameters. 
Those parameters are mostly determined in an ad-hoc or empirical manner.
This paper provides theoretical and quantitative tools to reason about the effects and trade-offs of these parameters on different metrics in the Nakamoto consensus including confirmation time, throughput, and fault tolerance.
We use these new tools to analyze and compare various altcoins and offer new insights and recommendations for setting parameters. 

Some new techniques developed in this paper may be of independent interests.
Assuming all block propagation delays are under $\Delta$ units of time, we show the arrivals of several species of honest blocks form {\em (ordinary) renewal processes~\cite{breuer2005introduction}}.  
That is, the inter-arrival times of such a process are independent and identically distributed (i.i.d.).  We show that the adversary must match the so-called double-laggers in order to succeed in any attack.  We derive the moment generating function (MGF) of their inter-arrival times in closed form.  This allows us to calculate quite accurately the probability that more double-laggers are mined than adversarial blocks over any confirmation time, which leads to a close {bound} 
of the fundamental latency--security trade-off.

Our main contributions include: 
{1)~The blockchain security question is reduced to a race between a renewal process and a Poisson process.}
2)~We provide an explicit formula for the security guarantee as a function of the latency.  This is equivalent to an upper bound on the latency that guarantees any desired security level.  
3)~By means of numerical analysis, the latency upper bound is shown to be close to a lower bound due to the private attack.
4)~We quantify how the block propagation delay bound, mining rates, and other parameters affect the latency--security trade-off. 
5)~We compare the performance and security of several prominent PoW longest-chain cryptocurrencies.
We note that several existing proofs in the literature 
are flawed (see Section~\ref{s:remarks} and the end of Section~\ref{sec:proof:safety}). 
In this paper, we carefully circumvent some common mistakes and develop a fully rigorous analysis.

The remainder of this paper is organized as follows.
Section~\ref{sec:model} reviews the Nakamoto consensus and describes a model for it.
Section~\ref{sec: main results} presents the main 
theorems. 
{Section~\ref{s:zero} presents a simple analysis of the zero-delay case.}
Section~\ref{s:consistency} develops a technical proof for the latency--security guarantee with delay.
Section~\ref{s:attack} calculates a bound due to the private attack.
Section~\ref{sec: numerical} is devoted to numerical results and discussions.
Section~\ref{sec: conclusion} concludes the paper.

\section{Preliminary and Model}
\label{sec:model}

\subsection{Review of the Nakamoto Consensus} 
\label{sec:model:informal}

The Nakamoto consensus centers around the PoW mechanism and the ``longest-chain-win'' rule.  The gist of the protocol can be described succinctly: At any point in time, every honest miner attempts to mine a new block that extends the longest blockchain in that miner's view; once a new block is mined or received,
the honest miner publishes or forwards it
through the peer-to-peer mining network.

If an honest miner mines a block at time $t$, the block must extend the miner's longest blockchain immediately before $t$.
The honest miner will also immediately publish
the block through a peer-to-peer network.
Under the $\Delta$-synchrony model, 
where $\Delta$ is a known upper bound on {block propagation delays},
all other miners will receive the block by time $t+\Delta$.
Note that $\Delta$ upper bounds the end-to-end delay between every pair of miners regardless of the number of hops between them.

In contrast, if an adversarial miner mines a block at time $t$, 
the block may extend any blockchain mined by time $t$ and may be presented to individual honest miners at any time after $t$. 


\subsection{Formal Model \DG{for the Nakamoto Consensus}}
\label{sec:model:formal}

In this subsection we present a very lean model for blockchain and mining processes which capture the essence of the Nakamoto consensus for our purposes.

\begin{definition}[Block and  mining]
\label{def:block}
    A genesis block, also referred to as block $0$, is mined at time $0$. Subsequent blocks are referred to as block $1$, block $2$, and so on, in the order they are mined.
    Let $T_b$ denote the time when block $b$ is mined.\footnote{A block in a practical blockchain system is a data structure that contains a unique identifier, a reference to its parent block, and some application-level data.
    The block number and mining time are tools in our analysis and are not necessarily included in the block. 
    The probability that two blocks are mined at the same time is zero in the continuous-time model. Nevertheless, for mathematical rigor, we can break ties in some deterministic manner even if they occur {(with no impact on the security--latency trade-off)}.}
\end{definition}

\begin{definition}[Blockchain and height]
Every block has a unique parent block that is mined strictly before it.  We use $\f{b} \in \{0,1,...,b-1\}$ to denote block $b$'s parent block number.
The sequence $\bl{0},\bl{1}, \dots, \bl{n}$ defines a blockchain if $\bl{0}=0$ and $\f{\bl{i}}=\bl{i-1}$ for $i =1,\dots ,n$.  It is also referred to as blockchain $b_n$ since $b_n$ uniquely identifies it. 
The height of both block $b_i$ and blockchain $b_i$ is said to be $i$.  
\end{definition}

Throughout this paper, ``by time $t$'' means ``during $(0, t]$''.

\begin{definition}[A miner's view]
\label{def:view}
    A miner's view at time $t$ is a subset of all blocks mined by time $t$.
    A miner's view can only increase over time. 
    A block is in its own miner's view from the time it is mined.
\end{definition}

\begin{definition}[A miner's longest blockchain]
    A blockchain is in a miner's view at time $t$ if all blocks of the blockchain are in the miner's view at time $t$.  
    A miner's longest blockchain at time $t$ is a blockchain with the maximum height in the miner's view at time $t$. Ties are broken in an arbitrary manner.\footnote{The Bitcoin protocol favors the earliest to enter the view.  How ties are broken has essentially no impact on the security--latency trade-off.}
\end{definition}

\begin{definition}[Honest and adversarial miners]
\label{def:miner}
    Each miner is either honest or adversarial.  A block is said to be honest (resp.\ adversarial) if it is mined by an honest (resp.\ adversarial) miner.
    An honest block mined at time $t$ must extend its miner's longest blockchain immediately before $t$.
\end{definition}


We assume all block propagation delays are upper bounded by $\Delta$ units of time in the following sense.
\begin{definition}[Block propagation delay bound $\Delta$]
\label{def:Delta}
    If a block is in any honest miner's view by time $t$, then it is in all miners' views by time $t+\Delta$.
\end{definition}

The adversary 
\DG{may use an} arbitrary strategy subject to \DG{(only)} the preceding constraints. Specifically, an adversary can choose to extend any existing blockchain. Once an adversarial block is mined, its miner can determine when it enters each individual honest miner's view subject to the delay bound $\Delta$ (Definition~\ref{def:Delta}).


This treatment cannot be fully rigorous without a well-defined probability space.  At first it appears to be intricate to fully described blockchains and an all-encompassing probability space.
One option (adopted in~\cite{dembo2020everything}) is to define blockchains as branches of a random tree that depend on the adversary's strategies as well as the network topology and delays. 
Other authors include in their probability space the random hashing outcomes as well as adversarial strategies (e.g.,~\cite{garay2015bitcoin}).
For our purposes it is 
sufficient (and most convenient) to include no more than the mining times of the honest and adversarial blocks in the probability space.
We show that barring a certain ``bad event''
in this probability space, blockchain consistency is guaranteed under all adversarial strategies and network schedules (under $\Delta$-synchrony).
Thus, the adversary's strategies and the network schedules do not have to be included in the probability space.

\begin{definition}[Mining processes]
\label{def:processes}
    Let $H_t$ (resp.\ $A_t$) denote the total number of honest (resp.\ adversarial) blocks mined during $(0, t]$.  We assume $(H_t, t\ge 0)$ and ($A_t, t\ge0)$ to be independent homogeneous Poisson point processes with rate $\alp$ and $\bet$, respectively.
\end{definition}


In lieu of specifying the number of honest and adversarial miners, \DG{this model is only concerned with} 
their respective aggregate \DG{(constant)} mining rates. 
This is in the same spirit as the permissionless nature of the Nakamoto consensus.
{In fact} the model and 
{techniques} apply to both centralized and decentralized mining.

Table \ref{table: notations} illustrates frequently used notations in this paper.

\begin{center}
\begin{table}[tb]
\begin{tabular}{ |c|l| } 
\hline
$\alp$ & {total} honest mining rate \\ 
$\bet $ & {total} adversarial mining rate \\
$\Delta$ & {block propagation delay upper bound} \\
$\A{s}{r}$ &  number of adversarial blocks mined during $(s,r]$\\
$\H{s}{r}$ &   number of honest blocks mined during   $(s,r]$ \\
$\M{s}{r}$ &  number of jumpers mined during   $(s,r]$ \\
$\X{s}{r}$ &  number of laggers mined during   $(s,r]$ \\
$\V{s}{r}$ &  number of double-laggers mined during   $(s,r]$ \\
$\Y{s}{r}$ &   number of loners mined during   $(s,r]$ \\
\hline
\end{tabular}
\caption{{Some frequently used} notations.}
\label{table: notations}
\end{table}
\end{center}

\section{Main Results} \label{sec: main results}

Throughout this section, the time unit is 
\DG{arbitrary} and fixed.
Let $\alp>0$ denote the total honest mining rate and let $\bet>0$ denote the total adversarial mining rate, both in blocks per unit of time.

\begin{definition}[Achievable security latency function]
\label{def:achievable}
    A blockchain 
    \DG{protocol} is said to {\em achieve the security latency function $\overline{\epsilon}(\cdot)$} if, for all $s>0$ and $t>0$,
    barring an event with probability no greater than $\overline{\epsilon}(t)$, all blocks that are mined by time $s$ and included in some honest miner's longest blockchain at time $s+t$ are also included in all honest miners' longest blockchains at all later times.
\end{definition}    
    
\begin{definition}[Unachievable security level for given latency]
\label{def:unachievable}
    We say a blockchain system cannot achieve the security level $\underline{\epsilon}(\cdot)$ 
    for a given latency $t>0$, if for some $s>0$, there exists an attack such that 
    some block that is mined by time $s$ and is included in a longest blockchain in some honest view at time $s+t$ is excluded from a longest blockchain in some honest view at some point in time after $s+t$ with probability $\underline{\epsilon}(t)$.
\end{definition}


The view adopted in Definitions~\ref{def:achievable} and~\ref{def:unachievable} is that an attack on the interval $[s,s+t]$ is successful if any honest miner ``commits'' a block by time $s+t$ and then commits a different block at the same height at some time $r\ge s+t$.
{We allow the adversary to mount} 
an attack tailored for $[s,s+t]$
{from as early as time 0}.
However, the interval of interest $[s,s+t]$ is exogenous, i.e., the adversary is not allowed to {adapt $s$ to their advantage}.


The Nakamoto consensus described in Section~\ref{sec:model} satisfies the following security--latency trade-offs.

\begin{theorem}
\label{th:zero achievable}
    Suppose $\bet<\alp$ and $\Delta=0$.  For every $t>0$, {the Nakamoto consensus} achieves the following security--latency function:
    \begin{align} \label{eq:e0}
      \overline{\epsilon}_0(t) =
      \left(1+\sqrt{\frac{\bet}\alp}\right)^2
      \exp\left( - \big( \sqrt{\alp}-\sqrt{\bet} \big)^2 t \right) .
    \end{align}
\end{theorem}

\begin{theorem} \label{th:zero_unachievable}
    For every latency $t>0$, {the Nakamoto consensus} cannot achieve the security level
    \begin{align} \label{eq:_e0}
        \underline{\epsilon}_0(t)
        =
        \sum_{k=0}^\infty p(k-1; \alp t, \bet t)
        \left( \frac{\bet}{\alp} \right)^k
        \left( 1 + k \left( 1 -  \frac{\bet}{\alp} \right) \right)
    \end{align}
    where $p(\cdot;\cdot,\cdot)$ denotes the probability mass function (pmf) of the Skellam distribution{~\cite{skellam1946frequency}}.
\end{theorem}

Theorems~\ref{th:zero achievable} and~\ref{th:zero_unachievable} have their counterparts in the case of non-zero propagation delays:

\begin{theorem}[Achievable security latency function]
\label{th:achievable}
    Suppose $\Delta>0$ and
    \begin{align} \label{e:beta<}
        \bet < \alp e^{-2\alp\Delta} .
    \end{align}
    Let 
    \begin{align}
        \eta(v)
        & = \frac{\alp v - v^2}
            {v^2 - \alp  v - \alp  v e^{(v-\alp)\Delta}
            + \alp^2 e^{2(v-\alp)\Delta} } \label{eq:phi}
    \end{align}
    and let $\theta$ denote the smallest positive zero of the denominator in~\eqref{eq:phi}.
    {The Nakamoto consensus} achieves the following security--latency function:
    \begin{align} \label{eq:minv}
        \overline{\epsilon}(t)
        = \min_{v \in (0, \theta)} c^2(v) \,
            e^{- (v-\eta(v)\bet) t}
    \end{align}
    where
    \begin{align}
    \begin{split}
        c(v)
        &= e^{\eta(v)\bet\Delta}
        \left(1-\frac{\bet}{\alp}e^{2\alp\Delta}\right)
        \frac{\eta(v)}{
        \frac1{1+\eta(\eta(v)\bet)}
        -
        \frac1{1+\eta(v)} } 
        .
    \end{split}
    \end{align}
\end{theorem}


{As $c(v)$ and $\eta(v)$ are infinitely differentiable functions, the optimization in~\eqref{eq:minv} is numerically easy.}
While the variable $v$ is optimized in~\eqref{eq:minv} for each latency $t$ of interest, by picking a maximum universal exponent we obtain a slightly weaker exponential security--latency function:
    \begin{align} \label{eq:ce^t}
        \overline{\epsilon}'(t) = c^2(u) \, e^{-(u-\eta(u)\bet)t}
    \end{align}
    where $u$ is a maximizer of the exponent, 
    i.e., $u-\eta(u)\bet\ge v-\eta(v)\bet$ for all $v\in(0,\theta)$.
{For many} practical parameters and typical security levels of interest, $\overline{\epsilon}(t)$ in~\eqref{eq:minv} is very close to an exponential function in $t$.  Hence~\eqref{eq:ce^t} provides an excellent approximation.

Theorems~\ref{th:zero achievable} and \ref{th:achievable} guarantee for all honest miners that all blocks in their longest blockchains that are received $t$ units of time earlier (hence, mined at least that much earlier) will always remain in all honest miners' longest blockchains in the future, except for a probability that is (essentially) exponentially small in $t$.
Unlike existing asymptotic results, these results provide a concrete and close upper bound on the probability of consistency violation under for a given confirmation time.
Equivalently, these results also upper bound the required confirmation time for every desired security level.

\begin{theorem}[Unachievable security lower for given latency]
\label{th:unachievable}
    Let
    \begin{align} \label{eq:xir}
        {
        \xi(\rho)
        =
        \frac{ (1-\rho)(\alp - \bet - \alp\bet\Delta) }
        { \alp - e^{(1-\rho)\bet\Delta} (\alp+\bet-\bet \rho) \rho}
        }
    \end{align}
    and
    \begin{align} \label{eq:qn}
        q(n) = \begin{cases}
        \xi(0) + \xi'(0) \quad & \text{if }\; n=0 \\
        \xi^{(n+1)}(0){/(n+1)!} \quad & \text{if }\; n=1,2,\dots
        \end{cases}
    \end{align}
    For every latency $t>0$, the {Nakamoto consensus} cannot achieve the following security level:
{ 
\begin{align} \label{eq:e2}
\begin{split}
    \underline{\epsilon}(t)
    &= 
    \left(1-\frac{\bet}{\alp}\right)
    e^{(\alp-\bet)t}
    \sum_{n=0}^\infty
    \sum_{\substack{k=0\\k+n>0}}^\infty
    q(n) \left(\frac{\bet}{\alp}\right)^k \\
    & \qquad \times F_1(k;\alp t) \, \overline{F}_2(t-(n+k)\Delta;n+k,\alp)
\end{split}
\end{align}
where $F_1(\cdot;\lambda)$ denote the cumulative distribution function (cdf) of a Poisson distribution with mean $\lambda$ and $\overline{F}_2(\cdot;n,\alp)$ denotes the complementary cdf of the Erlang distribution with shape parameter $n$ and rate parameter $\alp$.
}
\end{theorem}

Evidently, if the infinite sums in~\eqref{eq:_e0} and~\eqref{eq:e2} are replaced by partial sums for numerical evaluation, the resulting (tighter) security level remains unachievable.


\subsection{Remarks}
\label{s:remarks}

{Theorems~\ref{th:achievable} and~\ref{th:unachievable} assume the delay $\Delta>0$.  The bounds therein still apply if we set $\Delta=0$, but are slightly looser than the bounds in Theorems~\ref{th:zero achievable} and~\ref{th:zero_unachievable} for the zero-delay case.}

It is important to include the time of interest $s$ in Definitions~\ref{def:achievable} and~\ref{def:unachievable}.
The ``bad events'' 
for security breach
depend on $s$ as well as the latency $t$.  These well-defined events are concerned with block mining times, not how blocks form blockchains.\footnote{\DG{To be rigorous, we do not make} claims such as ``the blockchain/protocol/system satisfies 
consistency or liveness
\DG{properties}
with probability ...'' \DG{because} 
those properties themselves are not events in the probability space \DG{defined here}.}


We note that a number of previous analyses on the Nakamoto consensus assume a finite lifespan of the protocol~\cite{garay2015bitcoin, bagaria2019prism}, that is, a maximum round number is defined, at which round the protocol terminates.
The probability of consistency depends on the maximum round number.  In contrast, this paper does not assume a finite lifespan. Theorem~\ref{th:achievable} states that, barring a small probability event, confirmed blocks remain permanently in all miners' longest blockchains into the arbitrary future.

Even though we provide the same security guarantee for {\em every} blockchain after the confirmation latency $t$, {no one can simultaneously guarantee the same for {\em all} blocks that will ever be confirmed.  This is a simple consequence of Murphy's Law:}
If an adversary keeps trying new episodes of attacks,
with probability 1 a bad event will eventually occur to revert some confirmed honest blocks.

For technical convenience, 
we regard a block in a miner's longest blockchain to be confirmed after a certain amount of {\em time} elapses since the block is mined or enters the miner's view.  
Nakamoto~\cite{nakamoto2008bitcoin} originally proposed confirming a block after it is sufficiently {\em deep} in an honest miner's longest blockchain.
We believe both
confirmation rules are easy to use in practice. 
And the two confirmation rules imply each other 
in probability (see Appendix~\ref{sec:conversion} for further discussion). 

\begin{figure}
\centering
\includegraphics[width=\columnwidth]{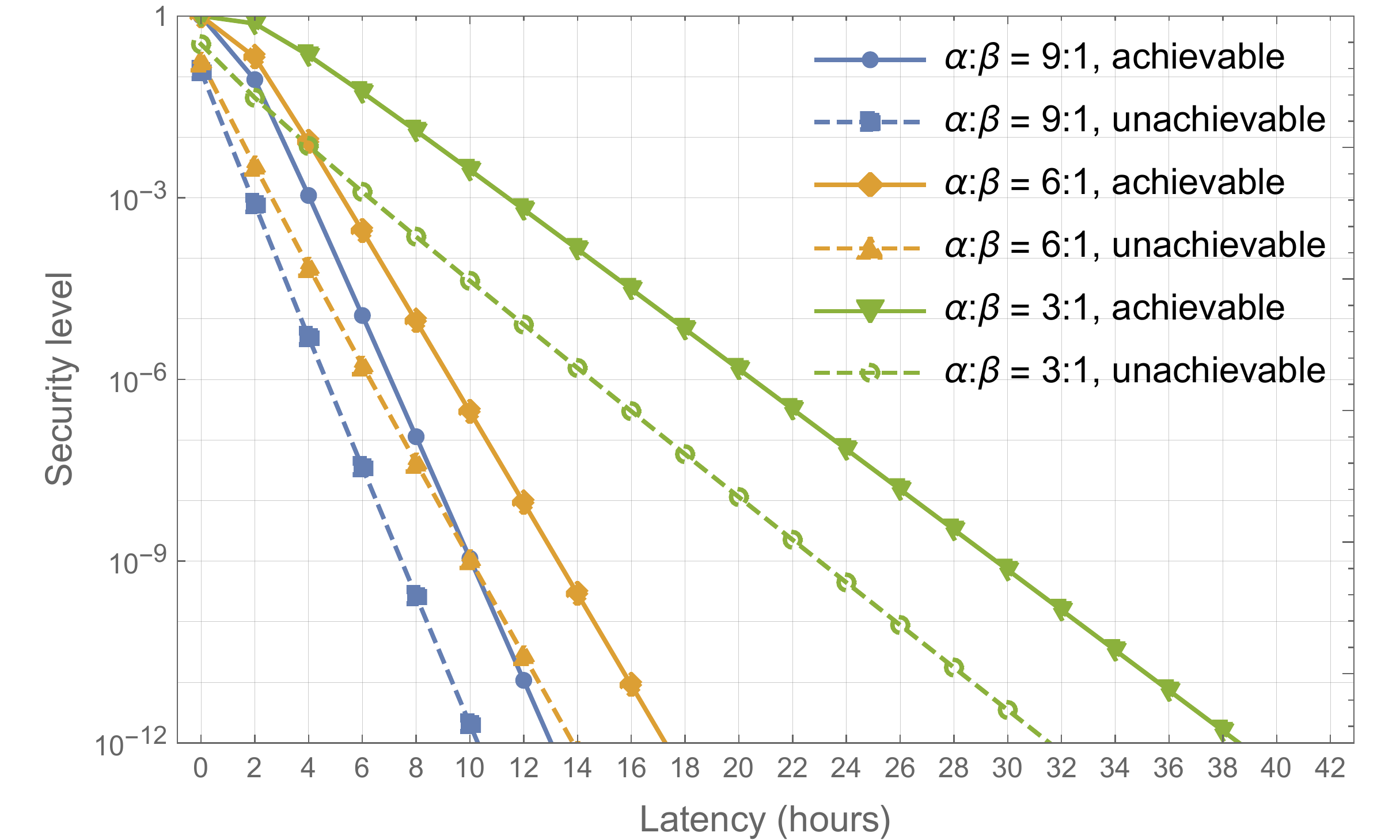}%
\caption{Bitcoin's latency--security trade-off with $\alp+\bet=1/600$ blocks per second and $\Delta=10$ seconds.}
\label{fig: security_latency}
\end{figure}

\subsection{Numerical Examples} 
The latency--security trade-off under several different sets of parameters is plotted in Figure~\ref{fig: security_latency}.  
The mining rate is set to Bitcoin's one block per $600$ seconds, or $\alp+\bet=1/600$ blocks/second. 
The propagation delay bound is assumed to be $\Delta=10$ seconds.
The latency upper and lower bounds are computed using Theorems~\ref{th:achievable} and~\ref{th:unachievable}, respectively.
In Figure~\ref{fig: security_latency}, 
all bounds appear to be exponential {for all but very small latency and high error probabilities}.
This implies the exponential bound~\eqref{eq:ce^t} is a good approximation of~\eqref{eq:minv} in Theorem~\ref{th:achievable} for the {typical} range of parameters of interest here.

It is instructive to examine concrete data points in Figure~\ref{fig: security_latency}:
If the adversarial share of the total network mining rate is 10\% ($\alp:\bet=9:1$), then a confirmation time of four hours is sufficient to achieve $10^{-3}$ security level, and a {ten-hour confirmation} achieves $10^{-9}$ security level.
These results are {about two hours away from} 
the corresponding lower bounds.
Also, for every additional hour of latency, the security improves by a factor of approximately 20.
If the adversarial share of the mining rate increases to 25\% ($\alp:\bet=3:1$), then 10 hours 40 minutes and 28 hours 45 minutes of confirmation times achieve $10^{-3}$ and $10^{-9}$ security levels, respectively, and the gap between the upper and lower bounds is {between five and seven} 
hours.
In general, the gap is proportionally insignificant at high security levels but can be otherwise at low security levels.
For given mining rates, the gaps {are similar at different}
security levels.
This indicates the lower bound~\eqref{eq:e2} is also approximately exponential with a {slightly steeper exponent than that of the upper bound}.


\subsection{The Mining Race}
\label{s:race}

A key proof technique in this paper is to reduce {the} blockchain security question to a race between the Poisson adversarial mining process and a counting process formed by a certain {species} of honest blocks.
In the zero-delay case, the latter is exactly the Poisson honest mining process.
{The race is analyzed in Section~\ref{s:zero}.}
In the bounded-delay case, we identify a {species} of honest blocks that are necessary for the adversary to match block-by-block.  We then upper bound the probability 
{that a matching can occur}
to yield an achievable security latency function (Theorem~\ref{th:achievable}) {in Section~\ref{s:consistency}}.
Conversely, we identify another {species} of honest blocks which are sufficient for the adversary to match to succeed {with a specific} attack.  We then lower bound the probability of such a match to yield the unachievability result (Theorem~\ref{th:unachievable}) {in Section~\ref{s:attack}}.



\section{The Zero-delay Case: The Race between Two Poisson Processes}
\label{s:zero}

In this section we analyze the fundamental security--latency trade-off of the Nakamoto consensus in the special case of zero block propagation delays.  
{
In this case, 
a private attack is known to be the worst among all possible attacks~\cite{dembo2020everything}.
While the gist of the security--latency trade-off is known, the thorough treatment here including both ``pre-mining'' and ``post-mining'' gains demonstrates some key techniques we shall use later in the much more complicated case with positive delay bound.
}


\subsection{Achievable Security Latency Function}
\label{s:zero delay achievable}

A sufficient condition for blocks mined by time $s$ and confirmed at time $r=s+t$ to remain permanent is that more honest than adversarial blocks are mined during every interval $(c,d]$ that covers $(s,r]$.
We let $F_{s,r}$ denote the ``bad event'' that
\begin{align} \label{eq:h<=a}
    H_{c,d}\le A_{c,d}
\end{align}
for some $c\in[0,s]$ and $d\in[r,\infty)$.  Note that~\eqref{eq:h<=a} is equivalent to
\begin{align}
  ( H_{c,s} - A_{c,s} )
  +
  ( H_{s,r} - A_{s,r} )
  +
  ( H_{r,d} - A_{r,d} )
  \le
  0 .
\end{align}
For convenience, let
\begin{align}
  M_- &= \min_{c\in[0,s]} \{ H_{c,s} - A_{c,s} \} \\
  M_+ &= \min_{d\in[r,\infty)} \{ H_{r,d} - A_{r,d} \}
\end{align}
where the adversary {attains} the maximum pre-mining {gain} if the attack began at the minimizer $c$, and the minimizer $d$ is the point of time where the adversary attains the maximum post-mining gain. 
It is easy to see that the bad event can be expressed as
\begin{align} \label{e:FHM}
  F_{s,r}
  =
  \{ M_- + H_{s,r} - A_{s,r} + M_+ \le 0 \} .
\end{align}
Because the $H$ and $A$ processes are independent and memoryless, the variables $M_-, H_{s,r}, A_{s,r}, M_+$ are mutually independent.

Note that $M_+$ is equal to the minimum of an asymmetric random walk where each step is $+1$ with probability $\alp/(\alp+\bet)$ and $-1$ with probability $\bet/(\alp+\bet)$.  It is well known that $-M_+$ is a geometric random variable with parameter $\bet/\alp$~\cite[Lemma 1]{janson1986moments}:
\begin{align} \label{eq:pm+}
  P(-M_+=k)
  =
  \left( 1-\frac{\bet}{\alp} \right) \left( \frac{\bet}{\alp} \right)^k, \qquad k=0,1,\dots.
\end{align}
Similarly, $M_-$ is equal to the minimum of the same random walk but with a finite duration $s$.  The distribution of $M_-$ converges to~\eqref{eq:pm+} as $s\to\infty$.  Evidently, $M_-$ is lower bounded by a random variable that is identically distributed as $M_+$.

We use the Chernoff's to upper bound the error probability: For every $\rho\in(0,\alp/\bet)$,
\begin{align}
  P(F_{s,r})
  &\le \expect{ \rho^{-(M_- + H_{s,r} - A_{s,r} + M_+)} } \\
  &\le \left( \expect{ \rho^{-M_+} } \right)^2
    \expect{ \rho^{A_{s,r}} } \expect{ \rho^{-H_{s,r}} } \\
  &= \left( \frac{ \alp - \bet }{ \alp - \bet \rho } \right)^2
    e^{\bet (r-s)(\rho-1)} e^{\alp (r-s)(\rho^{-1}-1)} \label{eq:F<=exp}
\end{align}
where we have used 
{the following} Laplace-Stieltjes transforms: 

\begin{lemma} \label{lem:poisson}
  Let $X$ be a Poisson random variable with parameter $q$.  For every $\rho>0$, we have
  \begin{align}
    \expect{ \rho^X }
    =
    e^{\lambda(\rho-1)} .
  \end{align}
\end{lemma}

{
\begin{lemma} \label{lem:geometric}
  Let $L$ be a geometric random variable with pmf $(1-q)q^k$, $k=0,1,\dots$.  For every $\rho>0$, we have
  \begin{align}
    \expect{ \rho^L }
    =
    \frac{1-q}{1-q\rho} .
  \end{align}
\end{lemma}
}

To obtain the tightest bound, we may optimize $\rho$ by setting the derivative with respect to $\rho$ to zero.  It is straightforward to show that the best $\rho$ is the unique real-valued root of a third order polynomial.
For simplicity, however, we use a specific choice
\begin{align}
  \rho = \sqrt{\frac{\alp}\bet}
\end{align}
which is asymptotically the best choice with $r-s\to\infty$ (it minimizes $\bet \rho+\alp/\rho$).  
The upper bound~\eqref{eq:F<=exp} then becomes~\eqref{eq:e0} in Theorem~\ref{th:zero achievable}.

\subsection{Unachievable Security for Given Latency}
\label{s:zero delay unachievable}

To examine how tight the achievable functions are, we shall compare them to unachievable security levels due to the following well-known attack:

\begin{definition}[Private attack with pre-mining and post-mining]
\label{def:private}
    Let $s>0$ and $t>0$ be given.  
    The adversary mines in private and publishes no blocks until it attains a longest blockchain in the entire system at or after time $s+t$, at which point it publishes this blockchain to conclude a successful attack.
    Specifically, until time $s$, the adversary always mines off the highest of all blocks in the system except for the highest honest block (which it tries to undermine).  From time $s$ onward, the adversary always mines off the tip of its longest private blockchain.
\end{definition}


Until time $s$, as soon as the adversary's private blockchain is more than one height behind, it starts afresh to mine off the parent of the highest honest block.  Let $L_r$, $0\le r\le s$, denote the pre-mining gain at time $r$, which is defined as the adversarial advantage relative to the parent of the highest honest block by time $r$.  It is not difficult to see that $(L_r, 0\le r\le s)$ is a birth--death process~\cite{ross1996stochastic} with constant birth rate $\bet$ and death rate $\alp$.  For all practical purposes $s$ is sufficiently large, so we regard the process to be in its steady state.  Hence the pmf of the pre-mining gain at time $s$ (denoted as $L$) is said to be the geometric distribution with parameter $\bet/\alp$ (the same as~\eqref{eq:pm+}).

Following similar arguments as in Section~\ref{s:zero delay achievable}, the adversary wins the race if the honest mining advantage during $(s,s+t]$ fails to counter the adversary's pre-mining and post-mining gains:
\begin{align} \label{eq:l1ham}
    1 + H_{s,s+t} - A_{s,s+t} \le L - M_+
\end{align}
where the ``1'' is due to the highest honest block mined by $s$.
The difference of the two independent Poisson random variables takes the Skellam distribution{~\cite{skellam1946frequency}}.  Also, the complementary cdf of the sum of i.i.d.\ geometric random variables $L$ and $-M_+$ takes the form of
    \begin{align} \label{eq:gn}
        \left( \frac{\bet}{\alp} \right)^n
        \left( 1 + n \left( 1 -  \frac{\bet}{\alp} \right) \right) .
    \end{align}
Hence the probability of the event~\eqref{eq:l1ham} can be rewritten as~\eqref{eq:_e0} in Theorem~\ref{th:zero achievable}.


\section{Achievable Security with Bounded Delay \DG{(Proof of Theorem~3.5)}}
\label{s:consistency}

As in the zero-delay case studied in Section~\ref{s:zero delay achievable}, we establish the achievable security latency function 
by studying a race between two mining processes.
Because delays cause forking, we only count 
a special {species} of honest blocks called \emph{loners}. 
A loner is an honest block that is not mined within $\Delta$ units of time of other honest blocks.
Let $b$ be a block mined {by} time $s$ and included in some honest miner's longest blockchain at time $s+t$. 
Section~\ref{sec:proof:safety} proves that block $b$ must be ``permanent'' if for all ${c} \leq s$ and ${d} \geq s+t$, more loners than adversarial blocks are mined during $({c,d}
]$. 

The race between the loner counting process and the adversarial mining process is difficult to analyze directly.  To make progress, we identify a renewal process as a surrogate of the loner process in Section~\ref{s:mgf}.  In Appendix~\ref{s:renewal-poisson} we develop a general formula to upper bound the probability that a renewal process loses the race against an independent Poisson process (this formula is useful in its own right).  In Section~\ref{s:2lagger-poisson} we invoke a special case of the formula to prove Theorem~\ref{th:achievable}.

For convenience and better intuition, we specifically choose the time unit to be equal to the block propagation delay bound in this proof.  
Hence $\Delta$ units of time in Theorem~\ref{th:achievable} becomes one (new) unit of time here.  This obviously normalizes the block propagation delay bound to 1 new unit of time.  Consequently,
the mining rate, aka 
the expected number of blocks mined \emph{per
new unit of time}, is equal to the expected number of blocks mined per maximum {allowable} delay.
With slight abuse of notion, we still use $\alp$ and $\bet$ as the mining rates under the new time unit.
At the end of the analysis we will recover Theorem~\ref{th:achievable} with an arbitrary time unit.

\subsection{Consistency Barring Loner Deficiency}
\label{sec:proof:safety}

\begin{definition}[Publication]
    A block is said to be published by time $t$ if it is included in at least one honest miner's view by time $t$.
    A blockchain is said to be published by time $t$ if all of its blocks are published by time $t$.
\end{definition}

\begin{definition}[$t$-credibility]
\label{def: c longest blockchain}
    A blockchain is said to be \emph{$t$-credible} if it is published by time $t$ and its height is no less than the height of any blockchain published by time $t-1$.
    If $t$ is unspecified, the blockchain is simply said to be credible (in context).
\end{definition}

There can be multiple $t$-credible blockchains, which may or may not be of the same {height}.
Once a block is published, it takes no more than 1 unit of time to propagate to all miners.  Hence at time $t$, an honest miner must have seen all blockchains published by $t-1$.  It follows that every honest miner's longest blockchain must be $t$-credible.  As we shall see, it is unnecessary to keep tabs of individual miner's views as far as the fundamental security is concerned.  Focusing on credible blockchains allows us to develop a simple rigorous proof with minimal notation.

\begin{definition}[Lagger] \label{def: lagger}
    An honest block mined at time $t$ is called a {\em lagger} if it is the only honest block mined during $[t-1, t]$.  By convention, the genesis block is honest and is regarded as the 0-th lagger.
\end{definition}

\begin{definition}[Loner]
    An honest block mined at time $t$ is called a {\em loner} if it is the only honest block mined during ${[t-1, t+1]}$.
\end{definition}


\begin{lemma} \label{lemma: d unique loner}
A loner is the only honest block at its height.
\end{lemma}

\begin{proof}
    Suppose block $b$ mined at time $t$ is a loner.
    By definition, no other honest block is mined during $[t-1, t+1]$. 
    By Definitions~\ref{def:view} and~\ref{def:Delta}, block $b$ is in all honest miners' views by time $t+1$.
    Thus, the height of every honest block mined after $t+1$ is at least ${\h{b}+1}$.
    Similarly, if an honest block is mined before $t-1$, its height must be smaller than $\h{b}$; otherwise, block $b$'s height would be at least $\h{b}+1$.  Hence, no other honest block has height $h(b)$.
\end{proof}





Suppose $0\le s < r$.
Let $\A{s}{r}=A_r-A_s$ denote the total number of adversarial blocks mined during $(s,r]$. 
Let $H_{s,r}= H_r-H_s$ denote the total number of honest blocks mined during time interval $(s,r]$.
Let $\X{s}{r}$ denote the total number of laggers mined during  $(s,r]$. 
Let $\Y{s}{r}$ denote the total number of loners mined during $(s,r]$. 
By convention, for every counting process $X$, we let $X_{s,r} = 0$ for all $s \ge r$.

\begin{lemma} \label{lemma: d Y < Z}
    Suppose $t \le r$.
    Let $s$ denote the mining time of the highest honest block shared by a $t$-credible blockchain and an $r$-credible blockchain.
    Then
    \begin{align} \label{eq: c Y < Z}
        \Y{s+1}{t-1} \le \A{s}{r}.
    \end{align}
\end{lemma}

\begin{proof}
    Let block $e$ denote the highest honest block shared by $r$-credible blockchain $d$ and $t$-credible blockchain $d'$ with $T_e = s$.
    Let block $b$ denote the highest block shared by blockchains $d$ and $d'$.
    Blocks $b$ and $e$ may or may not be the same block.
    \DG{An illustration of the} relationship between these blocks is 
    as follows:
    \begin{align} \label{eq:illustration}
        \begin{split}
        \setlength{\jot}{0pt}
        &\fbox{\color{white}a} -\cdots 
        -\cdots-\fbox{d} \;\; \\
        &\,\;|\quad\;  \qquad\qquad\qquad \text{time } r \\
        \fbox{e}-\cdots-&\fbox{b}-\cdots-\fbox{d'} \qquad\qquad \\
        \text{time } T_e=s \qquad & \qquad\qquad\qquad \text{time } t 
        \end{split}
    \end{align}

    If $t-s\le {2}$ or no loner is mined during $(s+1, t-1]$, obviously $ \Y{{s}+{ 1}}{t-{ 1}} = 0 \le \A{{s}}{r}$.
    Otherwise, consider loner $c$ mined during $({s}+{ 1}, t-{ 1}]$.
    We next show that $c$ can be paired with an adversary block mined during $(s,r]$.
    
    Since blockchain $e$ is $s$-credible and block $c$ is mined after time $s+1$, we have
    $\h{c} > \h{e}$. 
    Since blockchain $d$ is $r$-credible and blockchain $d'$ is $t$-credible, we have $\h{c} \le \min\{\h{d}, \h{d'}\}$.  {Consider the following two only possible cases:}
    \begin{enumerate}
    \item If $\h{e} < \h{c} \le \h{b}$, there exists at least one adversarial block at height $h(c)$ because all blocks in blockchain $d$ between block $e$ (exclusive) and block $b$ (inclusive) are adversarial by definition.
    \item If $\h{b} < \h{c} \le \min\{\h{d}, \h{d'}\}$,
    there is at least one adversarial block at height $h(c)$, because two divergent blockchains exist but loner $c$ is the only honest block at its height by Lemma \ref{lemma: d unique loner}.
    \end{enumerate}
    Thus, for {every} loner mined during $({s}+1, {t}-{1}]$, at least one adversarial block must be mined during $({s},r]$ at the same height.
    In particular, the adversarial block must be mined before $r$ because it is published by $r$.  Hence \eqref{eq: c Y < Z} must hold.
\end{proof}

We define {a ``bad event'':} 

\begin{definition}
    For all $s,r \ge 0$ and $\epsilon\in(0,1)$, let
    \begin{align}
    \F{s}{r} = \bigcup_{a\in [0,s], b \in [r,\infty)} \{\Y{a+1}{b-1-\epsilon} \le \A{a}{b}\} .
    \end{align}
\end{definition}

Let $\epsilon\in(0,1)$ be fixed for now.  We will eventually send $\epsilon\to0$.

\begin{lemma} \label{lemma: d consistency}
    Suppose block $b$ is mined by time $s$ and is included in a $t$-credible blockchain. Then, barring event $\F{s}{t}$, block $b$ is included in all $r$-credible blockchains for all $r\ge t$. 
\end{lemma}

\begin{proof}
    We first establish the result for $r\in[t,t+\epsilon]$ and then prove the lemma by induction.
    
    Fix arbitrary $r\in[t,t+\epsilon]$.  
    Let block $e$ denote the highest honest block shared by an $r$-credible blockchain and a $t$-credible blockchain that includes block $b$.  We have 
    \begin{align}
        \Y{T_e+1}{r-1-\epsilon}
        &\le
        \Y{T_e+1}{t-1} \\
        &\le
        \A{T_e}{r} \label{eq:YZTe}
    \end{align} 
    where~\eqref{eq:YZTe} is due to Lemma~\ref{lemma: d Y < Z}.
    Barring $\F{s}{t}$,
    \begin{align}
        \Y{a+1}{r-1-\epsilon} > \A{a}{r}
    \end{align}
    holds for all $a\in[0,s]$.
    Hence for~\eqref{eq:YZTe} to hold, we must have $T_e>s$.  Since $s\ge T_b$ by assumption, block $b$ must be included in blockchain $e$, which implies that block $b$ must also be included in the $r$-credible blockchain.
    
    Suppose the lemma holds for $r\in[t,t+n\epsilon]$ for some positive integer $n$.
    We show the lemma also holds for $r\in[r,t+(n+1)\epsilon]$ as follows:
    Let $t'=t+n\epsilon$.  
    If $\F{s}{t}$ does not occur, then $\F{s}{t'}$ does not occur either.
    Because block $b$ is included in a $t'$-credible blockchain, a repetition of the $r\in[t,t+\epsilon]$ case with $t$ replaced by $t'$ implies that block $b$ is included in all $r$-credible blockchains with $r\in[t',t'+\epsilon]$.  Hence lemma holds also for $r\in[t,t+(n+1)\epsilon]$.
    
    The lemma is then established by induction on $n$.
\end{proof}

Lemma~\ref{lemma: d consistency} guarantees that a block with some confirmation time is permanent/irreversible barring the bad event $\F{s}{r}$. It remains to upper bound $P(\F{s}{r})$ as a function of the confirmation time $r-s$.

Throughout this paper, the {arbitrary interval of interest $[s,r]$ is {\em exogenous}, so that both the honest and the adversarial} mining processes are Poisson {during the interval}.  We note that a recurrent mistake in the literature is to define a time interval according to some miners' views, actions or other observed outcomes. 
The boundaries of such an interval are in fact complicated 
{random variables} in an all-encompassing probability space.
Several prior analyses~(see, e.g., \cite{Ren2019Analysis, niu2019analysis, li2020continuous}) fail to account for the (subtle but crucial) fact that the mining processes are no longer Poisson when conditioned on {those random variables}.

\subsection{Some Moment Generating Functions}
\label{s:mgf} 

In this subsection we derive the MGF of several types of inter-arrival times.  These will be useful in the analysis of $P(\F{s}{r})$.

\begin{lemma} \label{lemma: d mgf R0}
    Let $Z$ be an exponential random variable with mean $1/\alp$.  Let the MGF of $Z$ conditioned on $Z\le1$ be denoted as $\phi_0(u)$.
    Then
    \begin{align} \label{eq:Phi0}
    \phi_0(u) = 
    \begin{cases}
        \frac{ \alp (1 -  e^{u-\alp}) }{ (1-e^{-\alp}) (\alp-u)} \quad &\text{ if } u\ne \alp\\
        \frac{\alp}{1-\ea} &\text{ if } u = \alp.
    \end{cases}
\end{align}
\end{lemma}


\begin{proof}
    The probability density function (pdf) of $Z$ conditioned on $Z\le1$ is simply
    \begin{align}
        \frac{\alp}{1-\ea} e^{-\alp z} 1_{\{0<z<1\}}
    \end{align}
    where $1_{\{\cdot\}}$ represents the indicator function which takes the values of 1 or 0 depending on whether the condition in the braces holds or not.
    The conditional MGF is thus 
\begin{align}
  \phi_0(u)
  & = \expect{ \ex{u Z} | Z\le1 } \\ 
  & = \int_0^1 \frac{\alp}{1-\ea} e^{-\alp z}\ex{uz} \text{d} z
\end{align}
which becomes~\eqref{eq:Phi0}.   
\end{proof}

\begin{lemma} \label{lemma: d mgf R1}
    Let $Z$ be an exponential random variable with mean $1/\alp$.  
    Let the MGF of $Z$ conditioned on $Z>1$ be denoted as $\phi_1(u)$.
    Then
    \begin{align} \label{eq:phi1}
      \phi_1(u) = \frac{\alp e^u}{\alp-u}
    \end{align}
    where the region of convergence is $u\in(-\infty,\alp)$.
\end{lemma}

\begin{proof}
    Conditioned on $Z>1$, the pdf of $Z$ is given by
    \begin{align}
        e^{\alp(-z+1)} 1_{\{z>1\}}.
    \end{align}
    The conditional MGF is thus: 
    \begin{align}
        \phi_1(u)
        & = \expect{ e^{uZ} | Z>1 } \\ 
        & = \int_1^{+\infty} e^{\alp(-z+1)}\ex{uz} \text{d} z.
    \end{align}
    The integral converges if and only if $u<\alp$, where the result is given by \eqref{eq:phi1}. 
\end{proof}

Recall the genesis block is the $0$-th lagger.  For $i=1,2,\dots$, let $X_i$ denote the time elapsed between the mining times of the $(i-1)$-st and the $i$-th lagger.  Let $K_i$ denote the number of honest blocks mined between the $(i-1)$-st lagger (excluded) and the $i$-th lagger (included). 

\begin{lemma} \label{LEMMA: LAGGER IID} 
    $(X_1,K_1), (X_2,K_2), \dots$ are i.i.d.
\end{lemma}

Lemma \ref{LEMMA: LAGGER IID} is proved in Appendix \ref{proof: lagger iid}.  {The lemma implies that laggers form an ordinary renewal process.}

\paragraph{Double-laggers.}
The loner process is not easy to characterize since whether a block mined at time $t$ is a loner depends not only on the past but also on future blocks (in $(t,t+1]$). 
In order to count loners, we examine a tightly-related species of honest blocks defined as follows.

\begin{definition}[Double-lagger]
    The first honest block mined after a loner is called a double-lagger.
\end{definition}

Note that a loner is also a lagger. 
So whenever 
two laggers are mined in a row, the former one is a loner and the latter one is a double-lagger.
As such, there is a one-to-one correspondence between loners and double-laggers.
We {shall} prove the independence of inter-double-lagger times and derive their MGFs, thus establishing the arrivals of double-laggers as a renewal process. 

Let $\V{s}{r}$ denote the total number of double-laggers mined during $(s,r]$. 
Let $\Vr{1}$ denote the time the first double-lagger arrives. Let $J_1$ be the number of laggers after the genesis block until the first double-lagger (included). For $i>1$, let $\Vr{i}$ denote the time elapsed between the $(i-1)$-st and the $i$-th double-lagger. Let $J_i$ be the number of laggers between the $(i-1)$-st double-lagger to the $i$-th double-lagger.

\begin{lemma} \label{lem:loner-double-lagger}
    For all $0 \leq s \le r$,
    \begin{align}
        \Y{s}{r} \ge \V{s}{r}-1.    
    \end{align}
\end{lemma}

\begin{proof}
    Because loners and double-laggers appear in consecutive pairs, every double-lagger mined during $(s,r]$ but the first corresponds to a distinct loner mined during $(s,r]$.
\end{proof}

\begin{lemma}\label{LEMMA: DOUBLE-LAGGER IID}
$(\Vr{1},J_1), (\Vr{2},J_2),\dots$ are i.i.d. 
\end{lemma}

Lemma \ref{LEMMA: DOUBLE-LAGGER IID} is proved in Appendix~\ref{s:doublelaggeriid}.


\begin{lemma} \label{lemma: lagger double-lagger dist}
    The time from a lagger to the next double-lagger follows the same distribution as an inter-double-lagger time.
\end{lemma}

\begin{proof}
    Let blocks $b$, $c$, $d$ be consecutive honest blocks.  Evidently, $T_c-T_b$ and $T_d-T_c$ are i.i.d.\ exponential random variables. Let $Q$ be the time elapsed from block $d$ to the next double-lagger. 
    If $d$ is a lagger, then $Q$ does not depend on whether $c$ is a lagger. Thus, for all $x$,
    \begin{align}
        P(Q\le x &\,|\,T_d-T_c>1) 
        = P(Q\le x \,|\, T_d-T_c>1, T_c-T_b>1) .
    \end{align}
    The left hand side is the cdf of the time between a lagger and the next double-lagger; the right hand side is the cdf of an inter-double-lagger time.  Hence the proof.
\end{proof}

For convenience we define the following function $g_\alp(u)$:
\begin{align} \label{eq:hau}
    g_\alp(u)
    =
    u^2 - \alp  u - \alp u e^{u-\alp} + \alp ^2 e^{2(u-\alp)}.
\end{align}
Evidently $g_\alp(u)>0$ if $u\le 0$ and $g_\alp(\alp)=0$.  Also, $g_\alp(u)$ is differentiable with bounded derivative on $[0,\alp]$.  
From now on, let $u_0$ denote the smallest zero of $g_\alp(\cdot)$, i.e.,
    $g_\alp(u_0)=0$,
and $g_\alp(u)\ne0$ for all $u\in[0,u_0)$. 
We must have $0<u_0\le \alp$.

\begin{lemma}\label{D MGF L}
    {The MGF of the inter-double-lagger time is}
    \begin{align}
        \phi(u)
        & = 1 + \frac{\alp u - u^2}
            {u^2 - \alp  u - \alp  u e^{(u-\alp )}
            + \alp ^2 e^{2(u-\alp)} } \label{eq:doublelaggermgf}
    \end{align}
    where the region of convergence is $(-\infty,u_0)$.
\end{lemma}

Lemma \ref{D MGF L} is proved in Appendix \ref{s:lem:mgf}.
The key is to study a Markov process: The initial state is a lagger.  With a known probability a double-lagger follows immediately to terminate the process.  With the remaining probability we visit a non-lagger state a geometric number of times until we return to the initial lagger state.  This allows us to write a recursion for the MGF of the inter-double-lagger time (aka the time till the double-lagger terminal state), the solution of which is~\eqref{eq:doublelaggermgf}.


\subsection{Race between the Double-Lagger Process and a Poisson Process}
\label{s:2lagger-poisson}

By Lemma~\ref{lem:loner-double-lagger}, it suffices to study the race between the double-lagger process and the adversarial mining process.
In Appendix~\ref{s:renewal-poisson}, we study the race between a general renewal process and a Poisson process to obtain the following result:

\begin{theorem}\label{th:renewal-poisson}
    {Let $(W_t)_{t\ge0}$ denote a renewal process where the mean of the renewal time is denoted by $m$.  Let the MGF of the renewal time be denoted by $\phi(u)$ where the region of convergence is $(-\infty,{\overline{u}})$, where ${\overline{u}}$ may be $+\infty$.}  Let $(A_t)_{t\ge0}$ denote a Poisson process with rate $\bet$.
    Let $\mu,\nu\ge0$ be fixed.  For $s,t>0$, the probability that there exists $c\in[0,s]$ and $d\in[s+t,\infty)$ such that the number of renewals in $(c,d]$ is at most the number of Poisson arrivals in $(c-\mu,d+\nu)$ plus $n$ is upper bounded:
    \begin{align} \label{eq:renewal-poisson}
    \begin{split}
        & P( \exists\, c\in[0,s], \, d\in[s+t,\infty) \text{ such that }
        W_{c,d} \le A_{c-\mu,d+\nu} + n ) \\
        &\le
        \exp((\phi(u)-1)\bet(\mu+\nu))
        \phi^{n+1}(u) \mathcal{L}^2( \phi(u) )
        \exp( - \psi(u) t )
    \end{split}
    \end{align}
    for all $u\in(0,{\overline{u}})$ where
    \begin{align} \label{eq:psi}
        \psi(u) = u + \bet - \bet\phi(u)
    \end{align}
    and
    \begin{align} \label{eq:laplace-stieltjes}
        \mathcal{L}(r) = \frac{ (r-1) ( 1-\bet\, m )}
        { r \left( \phi(\bet(r-1)) \right)^{-1} - 1 } .
    \end{align}
\end{theorem}

Theorem~\ref{th:renewal-poisson} applies to arbitrary renewal processes.  In this section, the renewal process of interest is the double-lagger process, where the mean and MGF of the inter-arrival time given by $m=e^{-2\alp}/\alp$ and $\phi(u)$ in~\eqref{eq:doublelaggermgf}, respectively.


\begin{lemma} \label{lemma: exist2}
    Define $\psi(u)$ using~\eqref{eq:psi} and~\eqref{eq:doublelaggermgf}.  If $\bet < \alp\ex{-2\alp}$, then there exists $u^*\in(0,u_0)$ such that 
    $\psi(u) > 0$ 
    for all $u\in (0, u^*]$.
\end{lemma}


\begin{proof}
    Evidently, $\psi(0) = 0$.  It is not difficult to calculate the derivative
\begin{align}
    \psi'(0)
    & = 1-\bet \phi'(0) = 1-\frac{\bet}{\alp\ex{-2\alp}}.
\end{align}
If $\bet < \alp\ex{-2\alp}$, we have $\psi'(0)> 0$.  By continuity, there must exist a $u^*<u_0$ such that $\psi'(u)>0$ and $\psi(u)>0$ for all $u\in (0,u^*]$.
\end{proof}

Let $u_1>0$ be the smallest positive number such that $\psi(u) = 0$.
Note that as $u\to u_0$, we have $\phi(u)\to \infty$ and $\psi(u) \to -\infty$.
By Lemma \ref{lemma: exist2}, 
$u_1$ exists and $u^* < u_1<u_0$.
{Indeed, Lemma~\ref{lemma: exist2} guarantees that the optimized exponent $\psi(u)$ is positive if $\bet<\alp e^{-2\alp}$.}

Theorem~\ref{th:achievable} is basically a special case of Theorem~\ref{th:renewal-poisson} where the renewal process is the double-lagger process.
Let $\mu=1$ and $\nu=1+\epsilon$.
Due to Lemma~\ref{lem:loner-double-lagger}, we given $n=1$ extra count to the Poisson process.
For convenience we let
\begin{align}
    \zeta(u)=\phi(u)-1.
\end{align}
The right hand side of~\eqref{eq:renewal-poisson} becomes
\begin{align}
\begin{split}
    e^{(2+\epsilon)\zeta(u)\bet}
    \left(1-\frac{\bet}{\alp}e^{2\alp}\right)^2
    \left( \frac{\zeta(u)}{
        \frac1{1+\zeta(\zeta(u)\bet)}
        -
        \frac1{1+\zeta(u)} } \right)^2 
                    e^{- (u-\zeta(u)\bet) t}
        .
    \end{split}
\end{align}
{It is interesting to note that the squared terms are the product of two identical multiplicative factors, one due to pre-mining and the other due to post-mining.}


Lastly, we recover Theorem~\ref{th:achievable} for the original arbitrary time unit.
In particular, the block propagation delays are bounded by $\Delta$ time units.
To reintroduce $\Delta$ into the result, we let $\tau=r\Delta$, $\sigma=s\Delta$, $\alp'=\alp/\Delta$, $\bet'=\bet/\Delta$, and $v=u/\Delta$.  These new variables and parameters are then defined under 
the original time unit. 
We define
\begin{align} \label{eq:Phieta}
    \eta(v) = \zeta(\Delta v) .
\end{align}
Plugging in~\eqref{eq:doublelaggermgf}, we have 
\begin{align}
    \eta(v)
    &= \frac{ \alp'v - v^2 }
        {v^2 - \alp' v - \alp' v e^{(v-\alp')\Delta}
        + \alp'^2 e^{2(v-\alp')\Delta} } .
   \label{eq:rho}
\end{align}

Suppose a block is mined by time $\sigma$ and is included in a $\tau$-credible blockchain. Applying Lemma~\ref{lemma: d consistency}, the block is included in all future credible blockchains barring {the} event $\F{\frac{\sigma}{\Delta}}{\frac{\tau}{\Delta}}$.
{
Using Theorem~\ref{th:renewal-poisson} with $\mu=1$, $\nu=1+\epsilon$, and $n=1$ with the time unit conversion and then letting $\epsilon\to0$,
an upper bound of 
$P(\F{\frac{\sigma}{\Delta}}{\frac{\tau}{\Delta}})$
is obtained as Theorem~\ref{th:achievable},
where $v_1$ is the smallest positive number such that $\psi(v_1)=0$.
It is easy to see that all blocks mined by $\sigma$ and included in a $\tau$-credible blockchain must be included in all credible blockchains thereafter barring the event
$\F{\frac{\sigma}{\Delta}}{\frac{\tau}{\Delta}}$,
whose probability is upper bounded by~\eqref{eq:minv} with the optimized parameter $v$.}
This conclusion is equivalent to the main theorem (with minor abuse of notation we still use $\alp$ and $\bet$ to replace $\alp'$ and $\bet'$, respectively, to follow some convention).
Thus, Theorem \ref{th:achievable} is proved.

\section{Liveness and Private Attack with Bounded Delay}
\label{s:attack}



In this section, we study liveness properties and generalize the analysis in Section~\ref{s:zero delay unachievable} to the case where block propagation delays are upper bounded.  {Following Section~\ref{s:consistency}, we let the time units be such that delay bound is 1 unit of time.}

\subsection{Blockchain Growth and Liveness}
\label{s:live}

\begin{definition}[Jumper]\label{def: jumper}
    We say the genesis block is the 0-th jumper.  After the genesis block, every jumper is the first honest block mined at least one maximum delay after the previous jumper.
\end{definition}



For $i=1,2,\dots$, let $M_i$ denote the time elapsed between the $(i-1)$-st jumper and the $i$-th jumper.  The following result is evident:

\begin{lemma} \label{lemma:inter-jumper}
    The inter-jumper times $M_1,M_2,\dots$ are i.i.d.\ and $M_i-1$ follows the exponential distribution with mean $1/\alp$.
\end{lemma}

Because $P(M_i>1)=1$ for all $i$, all jumpers have different heights (almost surely).
Let $\M{s}{r}$ denote the number of jumpers mined during time interval $(s,r]$.  For simplicity, we also assume that individual honest miners have infinitesimal mining power, so that almost surely no individual honest miner mines two blocks within one maximum delay.\footnote{{This is the worst case for honest miners.}}


Using Lemma~\ref{lemma:inter-jumper}, it is straightforward to establish the following result concerning the height of the longest (or credible) blockchains. 

\begin{lemma}[Blockchain growth]
    For all $s,t\ge0$, every honest miner's longest blockchain at time $s+t$ must be at least $n$ higher than every honest miner's longest blockchain at time $s$ with probability no less than
    \begin{align}
        F_2(t-1-n; n, \alp)
    \end{align}
    {where $F_2(\cdot;n,\alp)$ denotes the Erlang cdf with shape parameter $n$ and rate $\alp$.}
\end{lemma}

\begin{proof}
    If $t\le1$ then $F_2(t-1-n;n,\alp)=0$, so the lemma holds trivially.  We assume $t>1$.
    All jumpers mined during $(s,s+t-1]$ must be in every honest miner's views by $s+t$, where the first jumper is higher than the miner's longest blockchain at time $s$ and the last jumper is no higher than the miner's longest blockchain at time $s+t$.  Hence the probability of interest is no less than
    {
    \begin{align}
        P( J_{s,s+t-1} \ge n ) \ge P(J_{0,t-1}\ge n)
    \end{align}
    where the inequality is easily justified using a special case of Lemma~\ref{lem:dom}.  The event that $n$ or more jumpers are mined during $(0,t-1]$ is the same as that $n$ inter-jumper times can fit in a duration of $t-1$.  Hence
    \begin{align}
        P( J_{s,s+t-1} \geq n ) 
        &\ge P(M_1 + \dots + M_n \le t-1) \\
        &= P( (M_1-1) + \dots + (M_n-1) \le t-1-n ) \\
        &= F_2( t-1-n; n, \alp)        \label{eq:Fsn}
    \end{align}}%
    where~\eqref{eq:Fsn} is because $M_1-1, \dots, M_n-1$ are i.i.d.\ exponential random variables whose sum has the Erlang distribution.
\end{proof}

With the preceding techniques, it is also straightforward to establish the following {result} for blockchain quality or liveness.

\begin{lemma}[Blockchain liveness]
    Let $s>0$ and $t>1$.
    In every honest miner's longest blockchain at time $s+t$, the probability that $n$ or more of those blocks are honest blocks mined during $(s,s+t]$ is lower bounded by
    \begin{align} \label{eq:quality}
        \sum_{i=0}^{\infty}\ex{-\bet t}\frac{(\bet t)^i}{i!} F_2(t-i-n-1;i+n,\alp) .
    \end{align}
\end{lemma}

\begin{proof}
    Since $\M{s}{s+t-1}$ jumpers are mined during $(s,s+t-1]$ (with different heights), and at most $\A{s}{s+t}$ of them are matched by adversarial blocks, the number of surplus jumper blocks lower bounds the number of honest blocks in any honest miner's longest blockchain that are mined during $(s,s+t]$.  The said probability is thus lower bounded by
    \begin{align}
    & \P{ \M{s}{s+t-1} - \A{s}{s+t} \geq n } \notag \\
    & = \sum_{i=0}^{\infty} P( \A{s}{s+t} = i ) P(\M{s}{s+t-1} \ge i+n) \\
    & = \sum_{i=0}^{\infty}\ex{-\bet t}\frac{(\bet t)^i}{i!}
        P(M_1+\dots+M_{i+n} \le t-1)
    \end{align}
    which is equal to~\eqref{eq:quality} by~\eqref{eq:Fsn}.
\end{proof}

\subsection{Private Attack}

In addition to following the private attack of Definition~\ref{def:private}, the adversary {can also} manipulate the block propagation delays in the following manner: 1) Until time $s$, it presents all honest blocks to all honest miners with no delay, so that there is a unique blockchain consisting of honest blocks only.  {(This is not the most dangerous attack, but it is easy to analyze} and provides a bound on the unachievable security function).
2) From time $s$ onward, the adversary delays all honest blocks by the maximum allowed time $\Delta$.  
Since the adversary makes all delays zero before time $s$, the pre-mining gain $L$ {has the geometric pmf~\eqref{eq:pm+}}.
After time $s$, the first higher honest block is mined in exponential time, and the following honest blocks of increasing heights are mined with inter-jumper times.

The adversary wins the race if there exists $d\ge t$ such that the honest blockchain's growth during $(s,d-1]$ is no greater than $A_{s,d}+L-1$.  The growth of the honest blockchain during $(s,d-1]$ is identically distributed as the number of jumpers mined in $(0,d-s)$.  Hence the probability that the adversary wins is equal to the probability of
\begin{align}
  E_{s,t}
  &= \bigcup_{d\in[s+t,\infty)} \{ J_{0,d-s} \le A_{0,d-s} + L - 1 \} \\
  &= \bigcup_{d\in[t,\infty)} \{ J_{0,t} + J_{t,d} \le A_{0,t} + A_{t,d} + L - 1 \} .
\end{align}

The jumper counts $J_{0,t}$ and $J_{t,d}$ are in general dependent.  Using~\eqref{eq:dom-1} in Lemma~\ref{lem:dom}, we can replace $J_{t,d}$ by an identically distributed $J'_{0,d-t}$ which is independent of $J_{0,t}$ with a penalty of 1 count to yield
\begin{align}
  P(E_{s,t})
  &\ge P\bigg( \{J_{0,t} \le A_{0,t} + L - 1 \} \notag \\
  &\qquad \bigcup_{d\in(t,\infty)} \{ J_{0,t} + J'_{0,d-t} \le A_{0,t} + A_{t,d} + L - 2 \}
    \bigg) \\
  &= P( J_{0,t} \le A_{0,t} + N + L - 1 ) \label{eq:pejanl}
\end{align}
where
\begin{align}
  N = \max\left\{ 0, -1- \min_{r\in[0,\infty)} \{ J'_{0,r} - A'_{0,r} \} \right\}
\end{align}
lower bounds the post-mining gain, i.e., the maximum surplus the adversary can gain via private mining after time $s+t$.

We have a closed form expression for the Laplace-Stieltjes transform of $\min_{r\in[0,\infty)} \{ J'_{0,r} - A'_{0,r} \}$~\cite{kroese1992difference}.  That is, letting $\phi$ be $\phi_1$ given by~\eqref{eq:phi1} and $m=1+\frac1\alp$ in~\eqref{eq:laplace-stieltjes}, we can express the transform in the form of~\eqref{eq:xir}, the derivatives of which can be used to recover the pmf of $N$ as~\eqref{eq:qn}.
It is not difficult to show that the sum $A_{0,t}+L$ has the following pmf:
\begin{align}
    P( A_{0,t}+L = k )
    &= \sum^k_{i=0} 
    \left(1-\frac{\bet}{\alp}\right) \left(\frac{\bet}{\alp}\right)^{k-i}
    \ex{-\bet t}\frac{(\bet t)^i}{i!} \\
    &= \left(1-\frac{\bet}{\alp}\right) \left(\frac{\bet}{\alp}\right)^k
    \ex{(\alp-\bet) t} F_1(k;\alp t) \label{eq:alkf}
\end{align}
where $F_1(k;\alp t)$ stands for the cdf of a Poisson distribution with mean $\alp t$.
From~\eqref{eq:pejanl}, we have
\begin{align}
  P( E_{s,t} )
  &\ge
  \sum_{n=0}^{\infty} \sum_{k=0}^\infty  
  P(N=n, A_{0,t}+L=k)
  P( J_{0,t} \le n + k - 1 ) \\
  &=
  \sum_{n=0}^\infty P(N=n)
  \sum_{k=0}^\infty P(A_{0,t}+L=k) \notag \\
  &\qquad \qquad \times  ( 1 - F_2(t-n-k,n+k,\alp) ) .
  \label{eq:pef2}
\end{align}
Plugging~\eqref{eq:alkf} into~\eqref{eq:pef2} proves
Theorem~\ref{th:unachievable}, in which we also return the time units to the original ones.

\section{Numerical Analysis} \label{sec: numerical}

\subsection{Methodology}

\paragraph{Metrics.}
The performance metrics of a Nakamoto-style protocol include latency for a given security level, throughput, and fault tolerance (the upper limit of the fraction of adversarial mining in a secure system).
In this section, we
numerically compute the trade-off between different performance metrics of popular Nakamoto-style {cryptocurrencies} (Bitcoin Cash, Ethereum, etc.) and discuss their parameter selections. 

\paragraph{Block propagation delay.} 
The above metrics crucially depend on the block generation rate (or the total mining rate), maximum block size, and block propagation delay. 
The former two are explicitly specified in the protocol. 
The block propagation delay, 
however, depends on network conditions.
Block propagation delays in the Bitcoin network have been measured in~\cite{decker2013information,croman2016scaling, gencer2018decentralization}.
Such measurements are in general lacking for other cryptocurrencies.
It is observed in \cite{decker2013information} that there is a linear relationship between propagation delays and block size.
In this section, we assume the block propagation delay upper bound is determined by the block size  $S$ (in KB) according to the following formula:
\begin{align}
    \Delta = aS + b . \label{equ: linear}
\end{align}
We determine the coefficients $a$ and $b$ 
using propagation delay data from Bitcoin and Ethereum monitoring websites. 
In Bitcoin, the block size is about 1 MB.  
The propagation delay of Bitcoin blocks fluctuates over the years with an overall decreasing trend~\cite{bitcoin_time2};  the 90th percentile of block propagation is 4 seconds on average as of May 2021.  
Since $\Delta$ in our model needs to be an upper bound on propagation delay, we assume $\Delta=10$ seconds for a 1~MB Bitcoin block.  
According to~\cite{eth_time3}, the 90th percentile of Ethereum block propagation is between 1.5 and 1.75 seconds for an average block size of 25 KB. 
We round it up to 2 seconds for an upper bound. Using these data points, we estimate $a = 0.0098$ and $b = 0.208$ in formula \eqref{equ: linear}.

\begin{figure}
    \includegraphics[width=0.95\columnwidth]{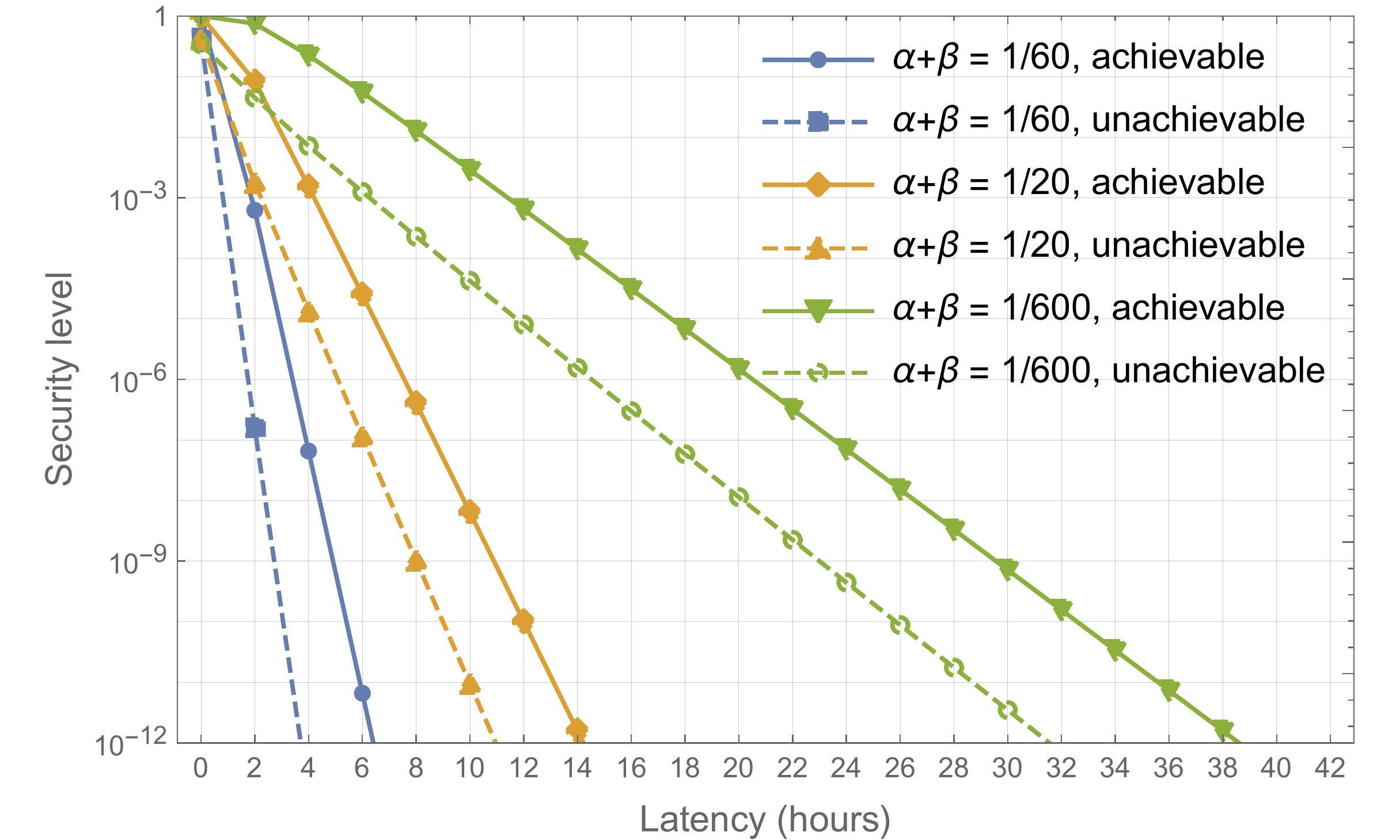}%
    \caption{Latency--security trade-off with $\Delta=10$ seconds and the percentage of adversarial mining power is 25\%.}
    \label{fig: security_latency2}
\end{figure}

\subsection{Confirmation time}

The latency--security trade-off 
{is} shown  in 
Figure \ref{fig: security_latency} {in Section~\ref{sec: main results}}
for the Bitcoin 
parameters.  
Figure~\ref{fig: security_latency2} illustrates how the trade-off changes if the block generation rate increases by 10 folds (to 1 block per minute), with everything else held the same.
The latency is much shorter under the higher block generation rate in this 
case. Also note that as the block generation rate further increases, the latency is longer. 

\begin{figure}
    \includegraphics[width=\columnwidth]{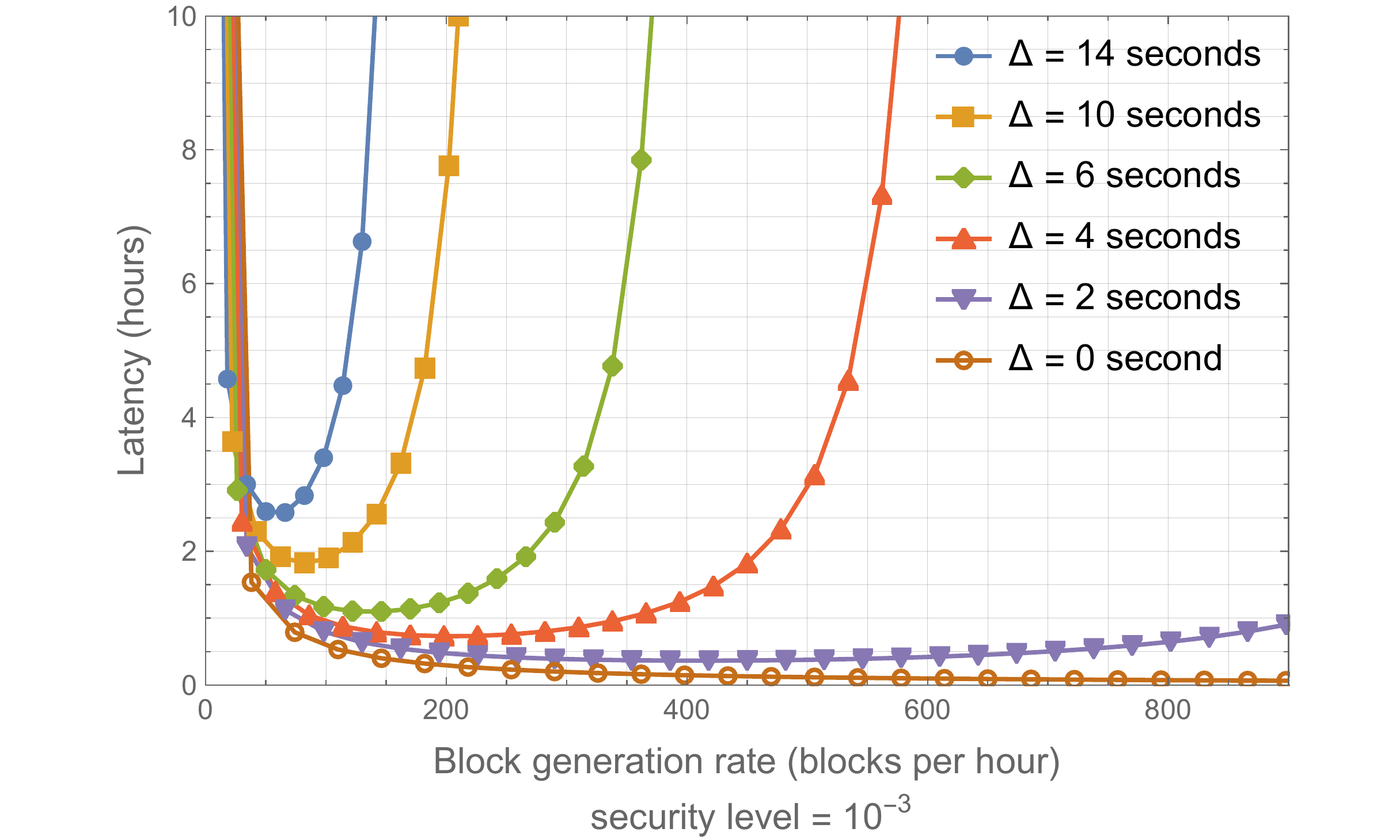}\\
    ~\\
    \includegraphics[width=\columnwidth]{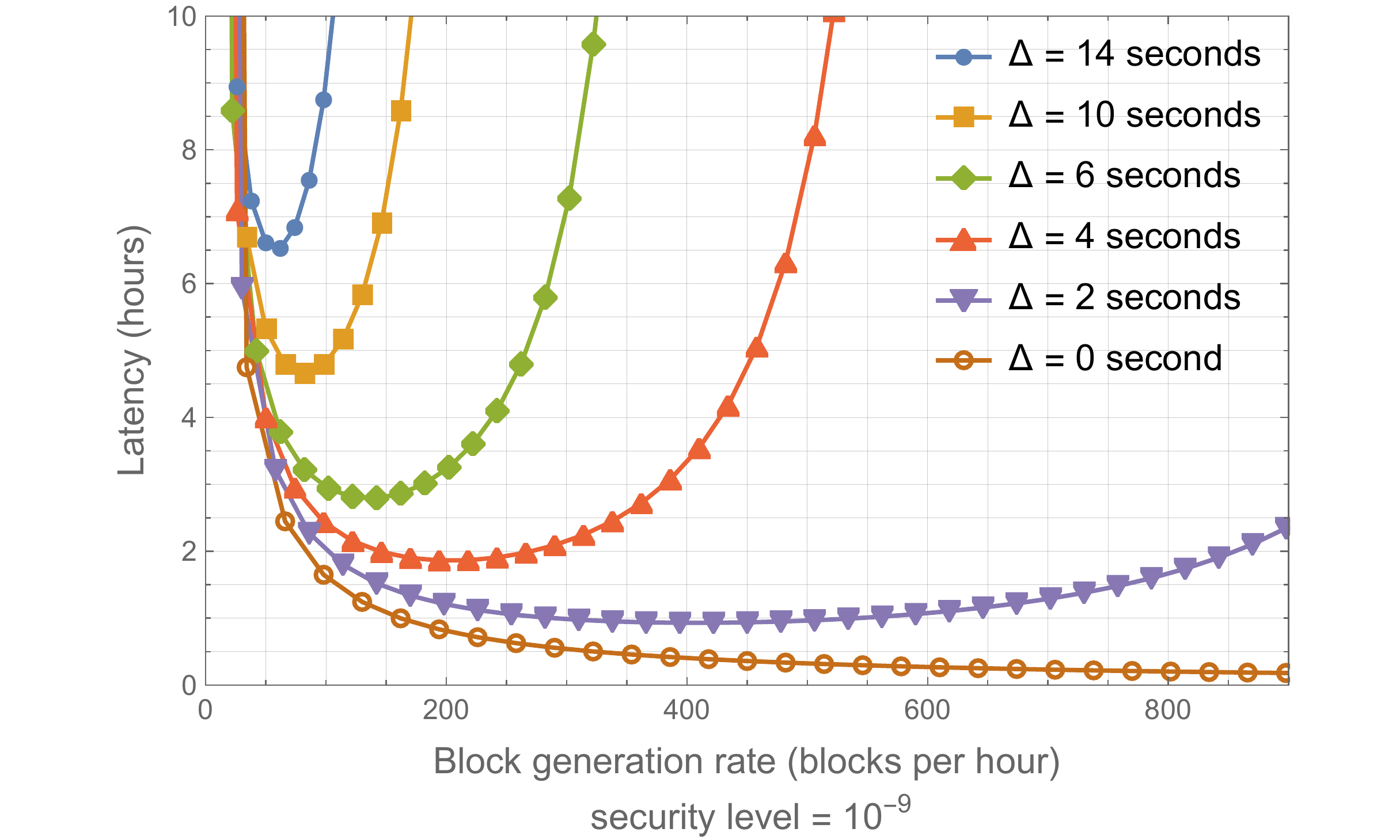}
    \caption{Latency required for different propagation delays. The percentage of adversarial mining power is 25\%.}
    \label{fig: latency_rate}
\end{figure}  

Figure \ref{fig: latency_rate}  illustrates sufficient confirmation times due to Theorem~\ref{th:achievable} for different security levels, block generation rates, and block propagation delays.  
The case of $\Delta=0$ is plotted using~\eqref{eq:e0} in Theorem~\ref{th:zero achievable}.
As expected, the latency is larger with longer block propagation delay and/or stronger security level requirement. 
Interestingly, increasing the block generation rate first reduces latency but eventually causes the latency to rise {rapidly (due to excessive forking)}.
From the graph, with the Bitcoin block propagation delay around 
{5 to 15} seconds, a sweet spot for block generation rate is between 50 and 200 blocks per hour in terms of optimizing latency.  

\begin{figure}
    \includegraphics[width=0.95\columnwidth]{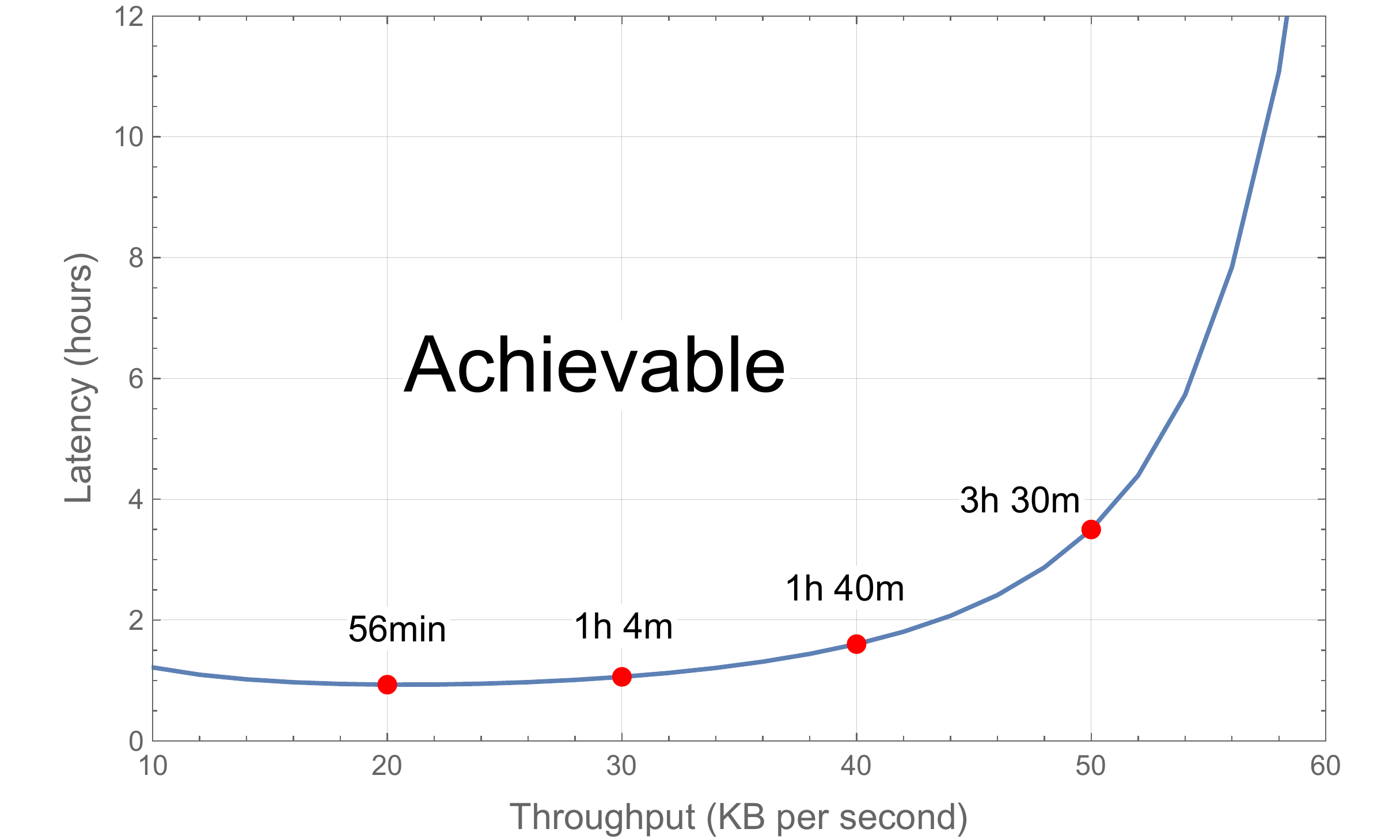}
    \caption{\jing{Achievable} latency for different throughput with \jing{$10^{-9}$} security level and 25\% adversarial mining power.}
    \label{fig: latency_throughput}
\end{figure}

\begin{table*}
\begin{tabular}{ |c|c|c|c|c|c|c|c|c|} 
\hline
\shortstack{Protocol \\ \wequ} & \shortstack{Maximum \\ block size} & \shortstack{Generation\\rate} & \shortstack{Propagation \\ delay (seconds)} & \shortstack{Latency: $10^{-3}$ \\ security level} &  \shortstack{Latency:$10^{-6}$  \\security level   } &  \shortstack{Latency: \jing{$10^{-9}$} \\ security level} & \shortstack{ \\ Throughput \\ (KB/second)} & \shortstack{ Fault \\ tolerance} \\ 
\hline
Bitcoin & 1 MB & 6 & 10 &  < 11h 39m &  < 20h 40m & < 29h 40m & 1.7 & 49.7\% \\
BCH & 8 MB & 6 & 78.5 & < 16h 55m & <  29h 59m &   < 43h 3m & 13.3 & 46.9\% \\
Litecoin & 1 MB & 24 & 10  &  < 3h 24m& <   6h 2m &  <  8h 39m  & 6.7 & 48.4\% \\
Dogecoin & 1 MB & 60 & 10 &  <  1h 56m &  <  3h 26m &  <  4h 56m & 16.6 & 46.2\%\\
Zcash & 2 MB & 48 & 19.8 & <  3h 40m &  <  6h 32m &  < 9h 23m & 26.7 & 44.2\% \\
Ethereum & 0.183 MB & 240 & 2 &  <  26m &  <  46m &  <  1h 6m & 12.2 & 46.9\%\\
\hline
\end{tabular}
\smallskip
\caption{Parameters and performances of Nakamoto-style Protocols. The percentage of adversarial mining power is 25\%. In formula \eqref{equ: linear}, \jing{a = 0.0098 and b = 0.208}.
}
\label{table: protocols}
\end{table*}

We recall that Theorem~\ref{th:achievable} requires condition~\eqref{e:beta<} to hold, i.e., the honest-to-adversarial mining ratio is bounded by
    $\bet/\alp < \ex{-2\alp \Delta}$,
where $\alp \ex{-2\alp \Delta}$ is the exact rate that loners are mined.
The actual sufficient and necessary condition for the consistency of the Nakamoto consensus in the infinite-latency limit is given in~\cite{dembo2020everything,gazi2020tight}:
\begin{align}
     \frac{\bet}{\alp} < \frac1{1+\alp\Delta} . \label{equ: max_rate2}
\end{align}
The product $\alp\Delta$ is equal to the average number of honest blocks mined per maximum delay period, which is usually set to a very small number to curb forking.
If $\alp\Delta\ll1$, as is the case in Bitcoin, the condition~\eqref{e:beta<} is almost identical to~\eqref{equ: max_rate2}.

\subsection{The latency--throughput trade-off} 



A larger block size may benefit throughput by carrying more transactions. On the other hand, the larger block size increases the propagation delay, which causes longer latency. A protocol designer may want to find a sweet spot that 
leads to the most desirable latency and throughput.

The throughput metric can be defined in a few different ways, ranging from the ``best-case'' throughput where the adversarial miners follow the protocol, to the ``worst-case'' throughput where the adversarial miners not only mine empty blocks but also use a selfish mining type of attack~\cite{eyal2014majority,garay2015bitcoin} to displace 
honest blocks. 
In this paper, we choose to focus on the ``best-case'' throughput, which 
is the throughput under normal operation and is 
perhaps what protocol designers have in mind when setting parameters. 
 


Figure \ref{fig: latency_throughput} illustrates the minimum latency required \jing{for a} given throughput according to Theorem~\ref{th:achievable} (it is actually a latency upper bound).
For several target throughput numbers (20, 30, 40, 50, and 60 KB per second), we also mark the corresponding \jing{latency bound}.  %
We see that we can achieve a latency less of around one hour with a throughput about 30 KB per second.
This is achieved by increasing the block generation rate and decreasing the block size.

\subsection{Case Studies in the Current Ecosystem}

Some PoW protocols attempt to 
better Bitcoin by increasing the block size (Bitcoin Cash) 
or the block generation rate (Litecoin). 
This subsection discusses the performance of these Bitcoin-like protocols.
Table \ref{table: protocols} describes the parameters, estimated propagation delay, and performances of the aforementioned protocols. 

\paragraph{Bitcoin Cash (BCH)} is a hard fork of Bitcoin 
from 2017.  BCH aims to increase the throughput by increasing the maximum block size to 8 MB while {keeping} 
the same block generation rate as Bitcoin~\cite{kwon2019bitcoin}.
As a result, the latency is increased from around \jing{30 hours to 43 hours for $10^{-9}$} security level. 
Had BCH increase the block generation rate (instead of the block size) by eight times, it would have obtained the same eight fold throughput improvement while at the same time shortened the latency by a factor of eight or so.

\paragraph{Litecoin}
 is also a fork of the Bitcoin Core client that dates back to 2011. Litecoin decreases the block generation time from 10 minutes to 2.5 minutes per block~\cite{litecoin_time1}.
For Litecoin, the latency is \jing{8 hours 39 minutes}, and the current throughput is 6.7 KB per second (for \jing{$10^{-9}$} security level and 25\% percentage of adversarial mining power). 

\jing{
\paragraph{Dogecoin}
is a cryptocurrency developed in 2013 as a joke, making fun of the wild speculation in cryptocurrencies. It is a fork of Litecoin blockchain with higher block generation rate (from 2.5 minute per block to one minute per block). The throughput of Dogecoin blockchain increases to 16.6 KB/second (compared with 6.7 KB/second and 1.7 KB/second for Litecoin and Bitcoin, respectively). Meanwhile, the $10^{-9}$ security-level confirmation time decreases to 5 hours (compared with 8.6 hours and 29.6 hours for Litecoin and Bitcoin, respectively). 
}


\paragraph{Zcash.}
Proposed in 2016, Zcash is aims to provide enhanced privacy features. 
In 2017, Zcash doubled the maximum block size from 1 MB to 2 MB~\cite{zcash_size_double}. Zcash also decreased the block interval from 10 minutes to 1.25 minutes~\cite{zcash_1}. 
Similar to Litecoin, ZCash can be improved by increasing the block generation rate (higher throughput) 
and/or decreasing block size (shorter latency).

\paragraph{Ethereum.}
The second largest cryptocurrency platform Ethereum has the block generation rate of 15 seconds per block \cite{vujivcic2018blockchain, eth_time1}. 
The maximum gas consumption for each Ethereum block is 12.5$\times 10^{6}$. Given that 21000 gas must be paid for each transaction and 68 gas must be paid for each non-zero byte of the transaction\cite{wood2014ethereum}, we estimate the maximum block size of an Ethereum block is 183 KB. 
Ethereum increases the block generation rate and decrease the block size. 
From Figure \ref{fig: latency_throughput}, for the throughput of around 12 KB per second, the latency bound is around \jing{1 hour  for $10^{-9}$ security level}, which is close to the current confirmation time of 1 hour and \jing{6} minutes. The Ethereum parameters appear to be well-chosen.

\paragraph{Summary.}
In general, most of Nakamoto-style cryptocurrencies start with Bitcoin as the baseline and aim to improve its throughput. 
Since Bitcoin has a very low block generation rate, the best option according to a principled method is to increase its block generation rate.
Additional improvements can be obtained by \emph{decreasing} the block size and \emph{further increasing} the block generation rate. 
This will not only increases throughput but also shortens the latency. 
Unfortunately, almost all of the cryptocurrencies we looked at either kept the block size the same as Bitcoin or went in the opposite direction to increase the block size, partly due to a lack of principled methodology. 
The only exception is Ethereum; Ethereum's parameters are very close to the optimal ones recommended in Figure~\ref{fig: latency_rate}.

\section{Conclusion and Future Directions}
\label{sec: conclusion}

{We have developed concrete bounds for the} latency--security trade-off for the Nakamoto consensus {under the assumption that the block propagation delays are upper bounded by a constant.  
If the total block generation rate (normalized by the block propagation delay bound) is not very high and the adversarial mining power is well below that of honest mining, the upper and lower bounds are quite close.
We have also}
applied the 
{new formulas} 
to analyze existing PoW longest-chain cryptocurrencies. 

Recent works~\cite{dembo2020everything,gazi2020tight} have established the 
{exact}
fault tolerance under high mining rate, but 
{concrete} bounds on latency {beyond the lower tolerance studied in this paper remain open.
Also, only asymptotic bounds are known for the Nakamoto consensus with dynamic participation and/or difficulty adjustment~\cite{garay2017bitcoin, chan2017blockchain}.  It would be interesting to establish concrete latency--security trade-offs in that more general setting}. 

\begin{acks}
We would like to thank David Tse and Sreeram Kannan for stimulating discussions about this work.
\end{acks}

\clearpage
\bibliographystyle{ACM-Reference-Format}
\bibliography{ref}

\appendix





\section{Converting Confirmation Depth to Confirmation Time}
\label{sec:conversion}

Let
\begin{align} \label{eq:kappa}
    \kappa(\lambda,\epsilon) 
    =
    \min \bigg\{ k: \;
    e^{-\lambda} \sum^\infty_{i=k} \frac{\lambda^i}{i!} \le \epsilon
    \bigg\} .
\end{align}
Then for every $\tau>0$, the probability that $\kappa( (\alp+\bet)\tau, \epsilon)$ or more blocks are mined in $\tau$ units of time is no greater than $\epsilon$.  For example, we learn from Figure~\ref{fig: security_latency} that \jing{4.35}  hours of latency guarantees a security level of 0.0005.  Using~\eqref{eq:kappa}, we obtain that \jing{ $\kappa(4.35\times6,0.0005)=45$}.  Hence if one counts {45} confirmation blocks, it implies that at least \jing{4.35} hours have elapsed with error probability 0.0005.  In all, \jing{45} confirmation blocks guarantees $10^{-3}$ security level, assuming that at most 10\% of the total mining power is adversarial.

At 10\% adversarial mining power, Nakamoto~\cite{nakamoto2008bitcoin} estimated that confirming after six blocks beats private attack at least 99.9\% of the time.  In contrast, \jing{45} confirmation blocks guarantees the same security level regardless of what attack the adversary chooses to employ.
We also note that while on average six blocks take only one hour to mine, with probability $10^{-3}$ it takes 2.75 hours or more to mine.

\section{Proof of Lemma \ref{LEMMA: LAGGER IID}} \label{proof: lagger iid} 
\label{s:lem:laggeriid}

For $i = 1, 2,  \dots$, let $Z_i$ denote the inter-arrival time between the $i$-th honest block and the $(i-1)$-st honest block throughout the appendices.  They are i.i.d.\ exponential random variables with mean $1/\alp$.

For convenience, we introduce the following shorthand notation within the bounds of this proof: 
\begin{align}
    L_n &= K_1+\dots+K_n\\
    l_n &= k_1+\dots+k_n  
\end{align}
for $n=1,2,\dots$.  By convention, we let $L_0=l_0=0$.
It is easy to see that
\begin{align}
    X_i =  Z_{L_{i-1}+1} + \dots+ Z_{L_i}
\end{align}
holds for $i=1,2,\dots$.  Also, the event that $K_i=k_i$ is equivalent to the event that
\begin{align}
    Z_{L_{i-1}+1} &\le 1, \dots,
    Z_{L_{i-1}+k_i-1} \le 1,
    Z_{L_{i-1}+k_i} > 1 . 
\end{align}
    Given $K_1=k_1, \dots, K_i = k_i$, the event $X_i\le x$ is equivalent to the event that
    \begin{align}
        Z_{l_{i-1}+1} + \dots+ Z_{l_i}\le x.
    \end{align}

    For all positive integers $n, k_1, k_2, \dots, k_n$  and real numbers $x_1, x_2, \dots, x_n$, we have
    \begin{align}
    & P(X_1\le x_1, K_1 = k_1, \dots, X_n\le x_n, K_n = k_n) \notag \\
    & =
    P(Z_1+\cdots+Z_{l_1}\le x_1, \notag \\
    & \qquad\qquad Z_1\le 1, \dots, Z_{l_1-1}\le 1, Z_{l_1}>1, \notag \\
    & \qquad\qquad \cdots, \notag \\
    & \qquad Z_{l_{n-1}+1}+\cdots+Z_{l_n}\le x_n, \notag \\
    & \qquad\qquad Z_{l_{n-1}+1}\le 1, \dots, Z_{l_{n}-1}\le 1, Z_{l_n}>1) \\
    & =
    P(Z_1+\cdots+Z_{k_1}\le x_1, \notag\\
    & \qquad\qquad Z_1\le 1, \dots, Z_{k_1-1}\le 1, Z_{k_1}>1 ) \notag \\
    & \qquad \times \cdots \times \notag \\
    & \quad~P( Z_{l_{n-1}+1}+\cdots+Z_{l_n}\le x_n, \notag \\
    & \qquad\qquad Z_{l_{n-1}+1}\le 1, \dots, Z_{l_n-1}\le 1, Z_{l_n}>1) 
    \label{d 17}
    \end{align}
    which is a product of $n$ probabilities, where~\eqref{d 17} is because $Z_1,Z_2,\dots$ are i.i.d.
    Let us define a two-variable function
    \begin{align}
        f( x, k ) = P( X_1 \le x, K_1=k )
    \end{align}
    for all $x\in(-\infty,\infty)$ and $k\in\{0,1,\dots\}$.
    The $i$-th probability on the right hand side of~\eqref{d 17} can be reduced as follows:
    \begin{align}
        & P(Z_{l_{i-1}+1}+\cdots+Z_{l_i}\le x_i, \notag\\
        & \qquad\qquad Z_{l_{i-1}+1}\le 1, \dots, Z_{l_i-1}\le 1, Z_{l_i}>1) \notag \\
        & = P(Z_1+\cdots+Z_{k_i}\le x_i,  \notag \\
        & \qquad\qquad Z_1\le 1, \dots, Z_{k_i-1}\le 1, Z_{k_i}>1)
        \label{eq:wliki} \\
        &= f( X_1 \le x_i, K_1=k_i )
    \label{eq:Pi}
    \end{align}
    for all $i= 1,\dots,n$, where~\eqref{eq:wliki} is because $Z_1,Z_2,\dots$ is stationary.
    Applying~\eqref{eq:Pi} to~\eqref{d 17} yields
    \begin{align}
        &P( X_1\le x_1, K_1=k_1, \dots, X_n\le x_n, K_n = k_n) \notag \\
        &=
        P(X_1\le x_1, K_1=k_1) \cdots P(X_1\le x_n, K_1=k_n) \\
        &=
        f( x_1, k_1  ) \cdots f(x_n, k_n) . \label{d 14}
    \end{align}
    Hence the joint probability distribution of $(X_i,K_i)^n_{i=1}$ decomposes and each term takes exactly the same form.  Thus Lemma \ref{LEMMA: LAGGER IID} is established.

\section{Proof of Lemma \ref{LEMMA: DOUBLE-LAGGER IID}}
\label{s:doublelaggeriid}

This proof takes the same form as the proof of Lemma~\ref{LEMMA: LAGGER IID}.
For convenience, we introduce the following shorthand within the bounds of this proof: 
\begin{align}
    M_n &= J_1+\dots+J_n \\
    m_n &= j_1+\dots+j_n  
\end{align}
for $n=1,2,\dots$.  By convention, we let $M_0=m_0=0$.
It is easy to see that
\begin{align}
    \Vr{i} = X_{M_{i-1}+1} + \cdots + X_{M_i}
\end{align}
holds for $i=1,2,\dots$.
Also, the event $J_i=j_i$ is equivalent to the event that 
\begin{align}
    K_{M_{i-1}+1}>1, \dots, K_{M_{i-1}+j_i-1}>1, K_{M_{i-1}+j_i}=1 .
\end{align}
Given $J_1=j_1,\dots,J_i=j_i$, the event $\Vr{i}\le w$ is equivalent to 
\begin{align}
    X_{m_{i-1}+1}+\dots+X_{m_{i}}\le w.
\end{align}
For all positive integers $n, j_1,\dots, j_n$ and real numbers $w_1,\dots, w_n$, we have 
\begin{align}
    & P(\Vr{1}\le w_1, J_1=j_1, \dots, \Vr{n}\le w_n, J_n = j_n) \notag \\
    & =  P(X_1+\dots+X_{j_1}\le w_1, \notag \\
    & \qquad \qquad K_1>1, \dots, K_{j_1-1}>1, K_{j_1}=1, \notag \\
    & \qquad \qquad \dots, \notag \\
    & \qquad X_{m_{n-1}+1} + \dots + X_{m_n}\le w_n, \notag \\
    & \qquad\qquad K_{m_{n-1}+1}>1,\dots, K_{m_n-1}>1, K_{m_n}=1
    ) \\
    & =  P(X_1+\dots+X_{j_1}\le w_1, \notag \\
    & \qquad \qquad K_1>1, \dots, K_{j_1-1}>1, K_{j_1}=1 ) \notag \\
    & \qquad \quad \times \cdots \times \notag \\
    & \quad~ P( X_{m_{n-1}+1} + \dots + X_{m_n}\le w_n, \notag \\
    & \qquad\qquad K_{m_{n-1}+1}>1,\dots, K_{m_n-1}>1,K_{m_n}=1)
    \label{d 75}
\end{align}
which is the product of $n$ probabilities, where~\eqref{d 75} is due to Lemma~\ref{LEMMA: LAGGER IID}, i.e., $(X_1,K_1),(X_2,K_2),\dots$ are i.i.d.  Moreover, the $i$-th probability on the right hand side of ~\eqref{d 75} can be reduced as:
\begin{align}
\begin{split}
   & P(X_{m_{i-1}+1} + \dots + X_{m_{i}}\le w, \\
    & \quad\qquad K_{m_{i-1}+1}>1, \dots,
     K_{m_{i}-1}>1,K_{m_i}=1)
        \\
    =~& 
    P(X_{1}+\dots+X_{j_i}\le w, \\
    & \quad \qquad K_1>1, \dots, K_{j_i-1}>1,K_{j_i}=1)
    \label{d 39}
\end{split}
\end{align} 
for all $i= 1,\dots,n$. Applying \eqref{d 39} to \eqref{d 75} yields
\begin{align}
\begin{split}
    & \wequ P(\Vr{1}\le w_1, J_1=j_1, \dots, \Vr{n}\le w_n, J_n = j_n) \\
& =P(\Vr{1}\le w_1, J_1=j_1)\cdots P(\Vr{1}\le w_n, J_1 = j_n). \label{d 76}
\end{split}
\end{align}
Hence the joint probability distribution of $(\Vr{i},J_i)_{i=1}^n$ decomposes and each term takes exactly the same function form. Thus Lemma~\ref{LEMMA: DOUBLE-LAGGER IID} is established.

\section{Proof of Lemma \ref{D MGF L}}
\label{s:lem:mgf}

By Lemma \ref{LEMMA: DOUBLE-LAGGER IID}, it suffices to consider $\Vr{1}$, the arrival time of the first double-lagger starting from time $0$.  Let $K$ denote the number of honest blocks until (including) the first lagger and let $b_1,\dots,b_K$ denote that sequence of blocks.  Then blocks $b_1, \dots, b_{K-1}$ are non-laggers, and block $b_K$ is a lagger (it may or may not be a double-lagger). 

    With probability $\ea$, $Z_1>1$. In this case, block $b_1$ is a double-lagger since the genesis block is a lagger. We know $K=1$ and $\Vr{1} = Z_1$. 

    With probability $1-\ea$, $Z_1 \le 1$. Then block $b_1$ is not a lagger. We have $Z_1\le 1, \dots, Z_{K-1}\le 1, Z_K>1$.  Let $I'$ denote the time from lagger $b_K$ to the next double-lagger.  Then we can write
    \begin{align} \label{d 48}
      \Vr{1} = 
      \begin{cases} 
          Z_1 &\quad \text{if } Z_1 > 1 \\
          Z_1 + \dots + Z_K + I' &\quad \text{if } Z_1 \le 1 \\
      \end{cases}
    \end{align}
    where $I'$ follows the same distribution as $\Vr{1}$ by Lemma \ref{lemma: lagger double-lagger dist}. 
    Thus the MGF of $\Vr{1}$ can be calculated as
    \begin{align}
    \expect{ e^{u\Vr{1}} } 
    &=  (1-\ea) \expect{ e^{u(Z_1+\dots+Z_K+I')}  \Big|  Z_1 \le 1 }
    \notag \\
    & \qquad + \ea \expect{ e^{uZ_1}  \Big|  Z_1 > 1 }\\
    &= 
    (1-\ea) \expect{ e^{u(Z_1+\dots+Z_K)} \Big| Z_1\le 1 } \expect{ e^{uI'} } \notag \label{d 57}\\
    & \qquad + \ea \expect{ e^{uZ_1}  \Big| Z_1 > 1 } \\
    &=  (1-\ea) \expect{ e^{u(Z_1+\dots+Z_K)} \Big| Z_1\le 1 } \expect{ e^{u\Vr{1}} }  \notag \\
    &  \qquad  + \ea \expect{ e^{uZ_1} \Big| Z_1>1 } \label{d 58}
    \end{align}
    where \eqref{d 57} is because $I'$ and $I_i$s are independent, and the fixed-point equation~\eqref{d 58} is because $I'$ is identically distributed as $\Vr{1}$.
    
    If
    \begin{align} \label{eq:cond1}
        (1-\ea) \expect{ e^{u(Z_1+\dots+Z_K)} \Big| Z_1\le 1 } <1
    \end{align}
    rearranging \eqref{d 58} yields
    \begin{align}
    \expect{ e^{u\Vr{1}} } 
    =
    \frac{ \ea \expect{e^{uZ_1} \Big| Z_1>1} }
    { 1 - (1-\ea) \expect{e^{u(Z_1+\dots+Z_K)} \Big| Z_1\le 1}} . \label{d 47}
    \end{align}
    We shall revisit the condition~\eqref{eq:cond1} shortly.
    
    Note that 
    \begin{align}
        P(K = k|Z_1\le 1) = (1-\ea)^{k-2} \ea, \quad k = 2,3,\dots.
    \end{align}
    Hence
    \begin{align}
        & 
        \expect{e^{u(Z_1+\dots+Z_K)} \Big| Z_1\le 1} \notag \\
        & = \sum_{k=2}^{\infty} P(K = k|Z_1\le 1) \times \notag\\ &\qquad\qquad\expect{\ex{u(Z_1+\dots +Z_k)} \Big| K=k, Z_1\le 1} \\
        & =  \sum_{k=2}^{\infty}(1-\ea)^{k-2}\ea \times \notag \\
        &\quad \expect{\ex{u(Z_1+\dots +Z_k)} \Big| Z_1\le 1, \dots, Z_{k-1}\le 1, Z_k > 1} \\
        & =  \sum_{k=2}^{\infty} (1-\ea)^{k-2}\ea \times \expect{e^{uZ_1} | Z_1\le 1}
          \times \cdots \times \notag \\ 
          & \quad \qquad \expect{e^{uZ_{k-1}} | Z_{k-1}\le 1} \times \expect{e^{uZ_k} | Z_k>1} \label{d 52}\\
        & =  \sum_{k=2}^{\infty}(1-\ea)^{k-2}\ea \,
        \phi^{k-1}_0(u)\phi_1(u) \label{d 53}\\
        & = \ea \phi_0(u)\phi_1(u) \sum_{k=0}^{\infty}(1-\ea)^{k}\phi^k_0(u)
    \end{align}
    where~\eqref{d 52} is due to mutual independence of inter-arrival times and \eqref{d 53} is due to Lemmas~\ref{lemma: d mgf R0}
    and~\ref{lemma: d mgf R1}.

    If
    \begin{align} \label{eq:cond2}
    (1-\ea)\phi_0(u)<1,
    \end{align}
    then the series sum converges to yield
    \begin{align} \label{d 55}
        \expect{e^{u(Z_1+\dots+Z_K)} \Big| Z_1\le 1}
        =
        \frac{\ea\phi_0(u)\phi_1(u)}{1-(1-\ea)\phi_0(u)} .
    \end{align}

    Let us examine the conditions~\eqref{eq:cond1} and~\eqref{eq:cond2}.
    Note that 
    \begin{align}
        (&1-\ea)\phi_0(u) =
        \begin{cases}\label{d 42}
        \frac{1-\frac{\ex{u}}{\ex{\alp}}}{1-\frac{u}{\alp}}  & \text{if } u \ne \alp, \\
        \alp & \text{if } u = \alp.
        \end{cases}
    \end{align}
    It is clear that~\eqref{d 42} is less than 1 for $u\le0$.  For $u>0$, because $e^u/u$ is monotone decreasing on $(0,1)$ and monotone increasing on $(1,+\infty)$, there must exist $u\le\alp$ that satisfies $(1-\ea)\phi_0(u)=1$.  Hence the region of convergence must be a subset of $(-\infty,\alp)$.
    
    
     Let $g_\alp(u)$ be defined as in~\eqref{eq:hau}.  Using~\eqref{d 55}, it is straightforward to show that
     \begin{align}
         (1-\ea) \expect{ e^{u(Z_1+\dots+Z_K)} \Big| Z_1\le 1 } = 1
     \end{align}
     is equivalent to $g_\alp(u)=0$.  Let $u_0>0$ be the smallest number that satisfies $g_\alp(u_0)=0$.  Then $u_0$ exists and $0<u_0\le\alp$ according to the discussion in Section~\ref{s:mgf}.  Also, $u<u_0$ implies~\eqref{eq:cond1}.  Therefore, the region of convergence for the MGF is $(-\infty,u_0)$.
  
    For $u<u_0$, we have by~\eqref{d 55} and Lemma~\ref{lemma: d mgf R1}:
    \begin{align}
    \expect{ e^{u\Vr{1}} } 
    & =  \frac{\ea\phi_1(u)(1-(1-\ea)\phi_0(u))}{1-(1-\ea)\phi_0(u) - \ea(1-\ea)\phi_0(u)\phi_1(u)} 
    \end{align}
    which becomes~\eqref{eq:doublelaggermgf}.

\section{Race Between a Renewal Process and a Poisson Process}
\label{s:renewal-poisson}

In this {appendix} we consider the race between a renewal process and a Poisson process.
Let $(W_t)_{t\ge0}$ denote the renewal process whose i.i.d.\ renewal times are denoted as $I_1,I_2,\dots$.
Let $(A_t)_{t\ge0}$ denote the Poisson process with rate $\bet$.
Let $[s,r]$ denote the interval of interest.

We give the Poisson process an early start by $\mu>0$ units of time and a late finish of $\nu\ge0$ units of time.  We also give the Poisson process an extra advantage of $n\ge0$.  We say the Poisson process wins the race if there exists $c\in(0,s]$ and $d\in(t,\infty)$ such that
\begin{align}
    W_{c,d} \le A_{c-\mu,d+\nu} + n ;
\end{align}
otherwise the renewal process $W$ wins the race.
The losing event is thus exactly:
\begin{align}
  F_{s,r}
  =
  \bigcup_{c\in[0,s], d\in[r,\infty)}
  \{ W_{c,d} \le A_{c-\mu,d+\nu} + n \} .
\end{align}

We divide the race to three segments: before $s$, from $s$ to $r$, and after $r$.  While the Poisson process is memoryless, the renewal process has memory, so the number of renewals in those three segments are dependent.
For convenience let us introduce two processes $(W'_t)_{t\ge0}$ and $(W''_t)_{t\ge0}$ which are i.i.d.\ as $(W_t)_{t\ge0}$.
Using the following result, we bound $P(F_{s,r})$ by studying a race with three independent segments.

\begin{lemma} \label{lem:dom}
    Let $d\ge r\ge s\ge0$.  The random variable $W_{s,d}=W_{s,r}+W_{r,d}$ is stochastically bounded:
    \begin{align}
    P( W_{s,r} + W'_{0,d-r} \le x-1 )
    &\le
    P( W_{s,r} + W_{r,d} \le x ) \label{eq:dom-1} \\
    &\le
    P( W_{s,r} + W'_{0,d-r} \le x ) \label{eq:dom} \\
    &\le
    P( W_{0,r-s} + W'_{0,d-r} \le x ) \label{eq:dom+1}
    \end{align}
    for every real number $x$.
\end{lemma}

We relegate the proof of Lemma~\ref{lem:dom} to Appendix~\ref{s:lem:dom}.

Using~\eqref{eq:dom} and~\eqref{eq:dom+1} yields
\begin{align} \label{eq:wwwa}
  P( F_{s,t} )
  \le
  P\left(
  \bigcup_{c\in[0,s], d\in[t,\infty)}
  \{ W'_{0,s-c} + W_{0,r-s} + W''_{0,d-r} \le A_{c-\mu,d+\nu} + n \}
  \right)
\end{align}
where we also use the fact that the adversarial process $A$ is independent of $W$, $W'$ and $W''$.
The inequality on the right hand side of~\eqref{eq:wwwa} is equivalent to
\begin{align}
  ( W'_{0,s-c} - A_{c-\mu,s-\mu} )
  +
  ( W_{0,r-s} - A_{s-\mu,r+\nu} - n )
  +
  ( W''_{0,d-r} - A_{r+\nu,d+\nu} )
  \le
  0
\label{eq:PF'} .
\end{align}
Let us define
\begin{align}
  M_- &= \min_{c\in[0,s]} \{ W'_{0,s-c} - A_{c-\mu,s-\mu} \} \\
  M_+ &= \min_{d\in[r,\infty)} \{ W''_{0,d-r} - A_{r+\nu,d+\nu} \} .
\end{align}
Using~\eqref{eq:PF'}, it is straightforward to check that the union event on the right hand side of~\eqref{eq:PF'} is identical to
\begin{align} \label{eq:mwam}
  M_- + ( W_{0,r-s} - A_{s-\mu,r+\nu} - n ) + M_+ \le 0 .
\end{align}
It is important to note that $W_{s,r}$, $A_{s-\mu,r+\nu}$, $M_-$, and $M_+$ are mutually independent.
We note that $M_-$ stochastically dominates $M_+$ as it is seeking the minimum gap between two processes over a shorter period of time.  For our convenience we introduce $M'_+$ to be i.i.d.\ with $M_+$.  Then,
\begin{align} \label{eq:pfpmm}
  P(F_{s,r})
  &=
  P( M_- + ( W_{0,r-s} - A_{s-\mu,r+\nu} - n ) + M_+ \le 0 ) \\
  &\le
  P( W_{0,r-s} \le A_{s-\mu,r+\nu} + n - M_+ - M'_+ ).
  \label{eq:pfpwamm}
\end{align}

Let the moment generating function of $W$'s renewal time be denoted as
\begin{align}
  \phi(u) = \expect{ e^{uI_1} } .
\end{align}
We then have that for every $u\ge0$ and positive integer $l$,
\begin{align}
  P( W_{0,r-s} \le l-1 )
  &\le P( I_1 + \dots + I_l \ge r-s ) \\
  &\le \expect{ \exp( u (I_1 + \dots + I_l - r+s) ) } \\
  &= ( \phi(u) )^l e^{-u(r-s)} .
  \label{eq:wstn}
\end{align}

The Laplace-Stieltjes transform of $M_+$, defined as
\begin{align}
    \mathcal{L}(\rho) = \expect{ \rho^{-M_+} }
\end{align}
is derived in~\cite{kroese1992difference} as~\eqref{eq:laplace-stieltjes} where $m=\expect{I_1}$ therein.
Plugging~\eqref{eq:wstn} into~\eqref{eq:pfpwamm}, we have
\begin{align}
    P(F_{s,r})  &\le \expect{ (\phi(u))^{ A_{s-\mu,r+\nu} + n - M'_+ - M_+ + 1 } } \\ 
    &=
    \expect{ (\phi(u))^{ A_{s-\mu,r+\nu} } }
    (\phi(u))^{ n + 1 }
    \left( \expect{ (\phi(u))^{ M_+ } } \right)^2 \\ 
    &\le \exp((\bet\phi(u)-\bet-u)(r-s))
     \exp((\phi(u)-1)\bet(\mu+\nu)) \notag \\
     & \qquad\qquad \times 
     ( \phi(u) )^{n+1} ( \mathcal{L}( \phi(u) ) )^2 
\end{align}
where we have used the Laplace-Stieltjes transforms of $M_+$ as well as 
$A_{s-\mu,r+\nu}$ which has a Poisson distribution with mean $\bet(r-s+\mu+\nu)$.
{Hence Theorem~\ref{th:renewal-poisson} is established.}

\section{Proof of Lemma~E.1} 
\label{s:lem:dom}

In this proof we make some simple intuitions precise.  Let $(W_t)_{t\ge 0}$ and $(W'_t)_{t\ge 0}$ denote two i.i.d.\ renewal processes.  Let $I_1,I_2,\dots$ and $I'_1,I'_2,\dots$ denote their respective i.i.d.\ renewal times.
We write
\begin{align}
    P( W_{s,t} + W_{t,d} \le x )
    &= \sum^\infty_{m=0} \sum^\infty_{n=0}
    P( W_{t,d} \le x-n, \, W_{0,s}=m, \, W_{s,t}=n ). \label{eq:wwwww}
\end{align}
The idea here is to divide the event into exclusive events corresponding to $m$ renewals in $(0,s]$ and $n$ renewals in $(s,t]$.  With $W_{0,s}=m$ and $W_{s,t}=n$, the first arrival after $t$ is the $(m+n+1)$-st renewal.  Regardless of the time from $t$ to this arrival, as long as $t+I_{m+n+2}+\dots+I_{m+x+1} > d$, we know that the $(m+x+1)$-st arrival is after $d$, so that at most $(m+x)-(m+n)=x-n$ renewals can occur during $(t,d]$, i.e., $W_{t,d}\le x-n$.  Hence~\eqref{eq:wwwww} leads to
\begin{align}
    P( 
    & W_{s,t} + W_{t,d} \le x ) \notag\\
    &\ge \sum^\infty_{m=0} \sum^\infty_{n=0}
    P( I_{m+n+2}+\dots+I_{m+x+1} > d-t, \, W_{0,s}=m, \, W_{s,t}=n ) .
    \label{eq:wwiiww}
\end{align}
Because $I_{m+n+2},\dots,I_{m+x+1}$ are renewal times after $t$, they are independent of $W_{0,s}$ and $W_{s,t}$.  Therefore,~\eqref{eq:wwiiww} is equivalent to
\begin{align}
    P(
    & W_{s,t} + W_{t,d} \le x ) \notag\\
    &\ge \sum^\infty_{m=0} \sum^\infty_{n=0}
    P( I'_1+\dots+I'_{x-n} > d-t, \, W_{0,s}=m, \, W_{s,t}=n ) \\
    &= \sum^\infty_{m=0} \sum^\infty_{n=0}
    P( W'_{0,d-t}\le x-n-1, \, W_{0,s}=m, \, W_{s,t}=n ) \\
    &= P( W_{s,t} + W'_{0,d-t} \le x-1 ) .
\end{align}
Hence the proof of~\eqref{eq:dom-1}.

To prove~\eqref{eq:dom}, we begin with~\eqref{eq:wwwww} and write
\begin{align}
    P( 
    & W_{s,t} + W_{t,d} \le x ) \notag\\
    &= \sum^\infty_{m=0} \sum^\infty_{n=0}
    P( t+Y+I_{m+n+2}+\dots+I_{m+x+1} > d, \, W_{0,s}=m, \, W_{s,t}=n ) \\
    &= \sum^\infty_{m=0} \sum^\infty_{n=0}
    P( t+Y+I'_2+\dots+I'_{x-n+1} > d, \, W_{0,s}=m, \, W_{s,t}=n ) 
    \label{eq:wwtyii}
\end{align}
where $Y$ denote the time between $t$ and the first renewal after $t$ and we have used the same aforementioned independence argument to arrive at~\eqref{eq:wwtyii}.  Conditioned on $W_{0,s}=m$ and $W_{s,t}=n$, the variable $Y$ is equal to part of the renewal time $I_{m+n+1}$.  Hence
\begin{align}
    P(Y\le y|W_{0,s}=m,W_{s,t}=n)
    &\ge P(I_{m+n+1} \le y|W_{0,s}=m,W_{s,t}=n) \\
    &= P(I'_1 \le y)
    \label{eq:yi'}
\end{align}
i.e., $Y$ is (conditionally) stochastically dominated by the renewal time.  In fact the intuition is very simple: If one has waited for some time, then the time to the next arrival should be shorter than the renewal time.
Using~\eqref{eq:wwtyii} and~\eqref{eq:yi'}, we have
\begin{align}
    P( 
    & W_{s,t} + W_{t,d} \le x ) \notag\\
    &\le \sum^\infty_{m=0} \sum^\infty_{n=0}
    P( t+I'_1+I'_2+\dots+I'_{x-n+1} > d, \, W_{0,s}=m, \, W_{s,t}=n ) \\
    &= \sum^\infty_{m=0} \sum^\infty_{n=0}
    P( W'_{0,d-t} \le x-n, \, W_{0,s}=m, \, W_{s,t}=n ) \\
    &= P( W_{s,t} + W'_{0,d-t} \le x ) .
\end{align}

\end{document}